\documentclass[preprint]{sigplanconf}
\usepackage[centertags,sumlimits,intlimits,namelimits,reqno]{amsmath}
\usepackage{latexsym,amsfonts,amssymb,exscale,enumerate,mathrsfs}
\usepackage[colorlinks,linkcolor=blue,citecolor=red,urlcolor=blue]{hyperref}
\usepackage{tikz}
\usetikzlibrary{backgrounds,circuits,circuits.ee.IEC,shapes,fit,matrix}
\usepackage[all,cmtip]{xy}
\usepackage{authblk}
\usepackage{stmaryrd}
\usepackage{amsthm}
\usepackage{anyfontsize}

\SetSymbolFont{stmry}{bold}{U}{stmry}{m}{n}

\pgfdeclarelayer{edgelayer}
\pgfdeclarelayer{nodelayer}
\pgfsetlayers{edgelayer,nodelayer,main}

\tikzset{font=\footnotesize}
\tikzstyle{none}=[inner sep=0pt]

\newcommand\z{{\mathbb Z}}
\renewcommand\k{\mathsf{k}}   
\newcommand\R{{\mathbb R}}
\newcommand\C{{\mathbb C}}
\newcommand\nn{{\mathbb N}}
\newcommand\bb{{\mathscr B}}
\newcommand{\idn}{\mathrm{id}}
\newcommand{\tw}{\mathrm{tw}}
\newcommand{\tm}{\tau}
\newcommand\pr{{\R[s,s^{-1}]}}
\newcommand\pk{{\k[s,s^{-1}]}}

\newcommand\maps{{\colon}}
\newcommand{\define}[1]{{\bf \boldmath #1}}
\newcommand{\Defeq}{\stackrel{\mathrm{def}}{=}}
\newcommand{\vectfun}{\theta}
\newcommand{\cospanfun}{\Theta}
\newcommand{\cospanfunrest}{\overline{\Theta}}
\newcommand{\ha}{\mathbb{HA}}
\newcommand{\ih}{\mathbb{IH}}
\newcommand{\ihcsp}{\ih^{\mathsf{Csp}}}
\newcommand{\ihcor}{\ih^{\mathsf{Cor}}}
\DeclareMathOperator{\mat}{\mathsf{Mat}}
\DeclareMathOperator{\vect}{\mathsf{Vect}}
\DeclareMathOperator{\ltids}{\mathsf{LTI}}
\DeclareMathOperator{\im}{im}
\DeclareMathOperator{\cospan}{\mathsf{Cospan}}

\DeclareMathOperator{\corel}{\mathsf{Corel}}
\DeclareMathOperator{\linrel}{\mathsf{LinRel}}
\DeclareMathOperator{\fmod}{\mathsf{FMod}}

\DeclareMathOperator{\bit}{bit}

\newcommand\addgen{\lower8pt\hbox{$\includegraphics[height=0.7cm]{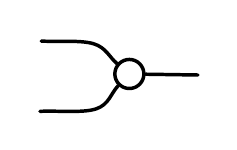}$}}
\newcommand\zerogen{\lower5pt\hbox{$\includegraphics[height=0.5cm]{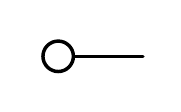}$}}
\newcommand\copygen{\lower8pt\hbox{$\includegraphics[height=0.7cm]{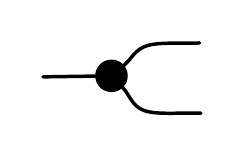}$}}
\newcommand\discardgen{\lower5pt\hbox{$\includegraphics[height=0.5cm]{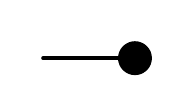}$}}
\newcommand\delaygen{\lower6pt\hbox{$\includegraphics[height=0.6cm]{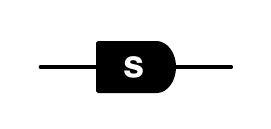}$}}
\newcommand\minonegen{\lower6pt\hbox{$\includegraphics[height=0.6cm]{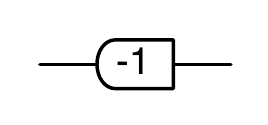}$}}
\newcommand\delayopgen{\lower6pt\hbox{$\includegraphics[height=0.6cm]{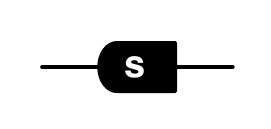}$}}
\newcommand\scalargen{\lower6pt\hbox{$\includegraphics[height=0.6cm]{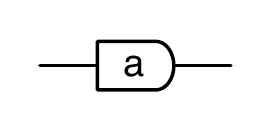}$}}
\newcommand\addopgen{\lower8pt\hbox{$\includegraphics[height=0.7cm]{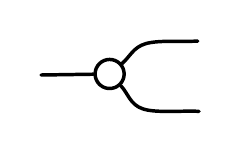}$}}
\newcommand\zeroopgen{\lower5pt\hbox{$\includegraphics[height=0.5cm]{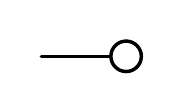}$}}
\newcommand\copyopgen{\lower8pt\hbox{$\includegraphics[height=0.7cm]{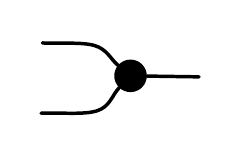}$}}
\newcommand\discardopgen{\lower5pt\hbox{$\includegraphics[height=0.5cm]{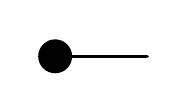}$}}
\newcommand\scalaropgen{\lower6pt\hbox{$\includegraphics[height=0.6cm]{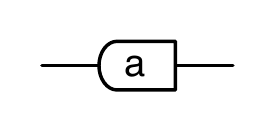}$}}
\newcommand\delaygenl{\lower6pt\hbox{$\includegraphics[height=0.6cm]{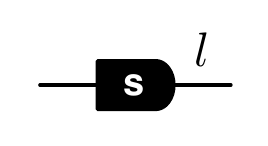}$}}
\newcommand\delayopgenl{\lower6pt\hbox{$\includegraphics[height=0.6cm]{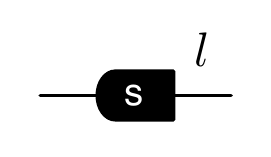}$}}
\newcommand\delaygenk{\lower6pt\hbox{$\includegraphics[height=0.6cm]{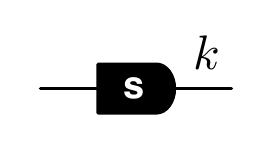}$}}
\newcommand\delayopgenk{\lower6pt\hbox{$\includegraphics[height=0.6cm]{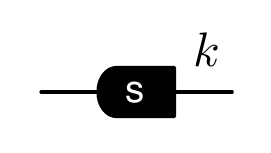}$}}
\newcommand\twist{\lower6pt\hbox{$\includegraphics[height=0.6cm]{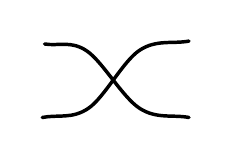}$}}
\newcommand\id{\lower3pt\hbox{$\includegraphics[height=0.3cm]{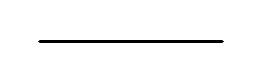}$}}
\newcommand\syntax{\mathbb{S}}

  \makeatletter
     \def\moverlay{\mathpalette\mov@rlay}
     \def\mov@rlay#1#2{\leavevmode\vtop{%
        \baselineskip\z@skip \lineskiplimit-\maxdimen
        \ialign{\hfil$#1##$\hfil\cr#2\crcr}}}
   \makeatother
\newcommand\twarr[2]{%
\mathrel{\mathop{\moverlay{\scriptstyle\xrightarrow{\,#1\,}\cr{\lower.2em\hbox{$\scriptstyle{}_{#2}$}}}}}}
\newcommand{\dtrans}[2]{\hbox{$\;\twarr{#1}{#2}\;$}}
\newcommand{\labelSep}{\,}

\newtheorem{clm}{Claim}[section]
\newtheorem{corollary}[clm]{Corollary}
\newtheorem{proposition}[clm]{Proposition}
\newtheorem{definition}[clm]{Definition}
\newtheorem{lemma}[clm]{Lemma}
\newtheorem{theorem}[clm]{Theorem}

\newtheorem{exm}[clm]{Example}

\newtheorem{remark}[clm]{Remark}

\begin{document}   

\title{A categorical approach to open and interconnected dynamical systems}
\authorinfo{Brendan Fong}{Department of Computer Science \\ University of
  Oxford, and
\\ Department of Mathematics \\ University of Pennsylvania}{}
\authorinfo{Pawe{\l} Soboci{\'n}ski}{School of Electronics and \\ Computer
  Science\\
University of Southampton}{}
\authorinfo{Paolo Rapisarda}{School of Electronics and \\ Computer Science\\
University of Southampton}{}

\maketitle

\begin{abstract}
In his 1986 Automatica paper Willems introduced the influential behavioural
approach to control theory with an investigation of linear time-invariant (LTI)
discrete dynamical systems. The behavioural approach places open systems at its
centre, modelling by ‘tearing, zooming, and linking’. In this paper, we show
that these ideas are naturally expressed in the language of symmetric monoidal
categories. 

Our main result implies an intuitive sound and fully complete string diagram
algebra for reasoning about LTI systems. These string diagrams are closely
related to the classical notion of signal flow graphs, but in contrast to
previous work, they endowed with semantics as multi-input multi-output
transducers that process streams with an \emph{infinite past} as well as an
infinite future. At the categorical level, the algebraic characterisation is
that of the prop of corelations, instead of a pushout of the props of spans and
cospans. 

Using this framework, we then derive a novel structural characterisation of
controllability, and consequently provide a methodology for analysing
controllability of networked and interconnected systems. We argue that this
methodology has the potential of providing elegant, simple, and efficient
solutions to problems arising in the analysis of systems over networks, a
vibrant research area at the crossing of control theory and computer science. 
\end{abstract}

\section{Introduction}

The remit of this paper is the development of a sound and fully complete
equational theory of linear time-invariant (LTI) dynamical systems. As in
previous work on linear systems~\cite{BSZ1,BSZ2,BSZ3,BE,Za}, the terms of this
theory---representing the LTI systems themselves---are best represented as
diagrams that closely resemble signal flow graphs~\cite{Sh}. 

Before exploring the underlying mathematics, we illustrate our contribution by
considering the behaviour of the signal flow graph below, rendered in
traditional, directed notation.
\begin{equation}\label{eq:examplesfg}
\lower12pt\hbox{$\includegraphics[height=1.5cm]{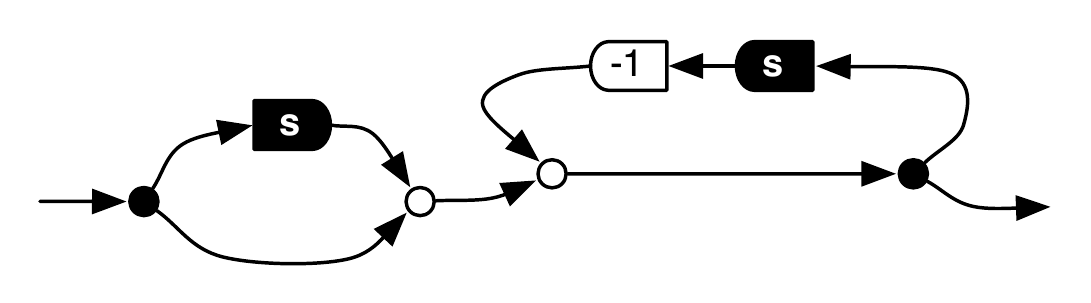}$}
\end{equation}
The circuit inputs a stream of values from a field $\k$, e.g.\ the rational
numbers, and outputs a stream of values. All processing is done synchronously
according to a global clock.  The white circles are adders, the black circles
are duplicators, the `$s$' gates are 1-step delays and the `$-1$' gate is an
instance of an amplifier that outputs $-1$ times its input.

We can express \eqref{eq:examplesfg} as a string diagram, in the sense of Joyal
and Street \cite{JS}, by forgetting the directionality of wires and composing
the following basic building blocks using the operations of monoidal categories.
\begin{multline*}
\copygen \mathrel{;} 
\!\!\!\!\!\!\!
\begin{array}{c} \delaygen \\ \oplus \\ \id \end{array} 
\!\!\!\!\!\!\!\mathrel{;}\!\! \addgen
\mathrel{;} \!\!\!\!\!\!\! \begin{array}{c} \discardopgen \\ \oplus \\ \id \end{array} 
\!\!\!\!\!
\mathrel{;} 
\!\!\!\!\!
 \begin{array}{c} \copygen \\ \oplus \\ \id \end{array} 
\!\!\!\!\!
\mathrel{;}
\\
\!\!\!\!\! 
 \begin{array}{c} \minonegen \\ \oplus \\ \id \end{array} 
\!\!\!\!\! 
 \mathrel{;}
\!\!\!\!\! 
 \begin{array}{c} \delayopgen \\ \oplus \\ \id \end{array} 
 \!\!\!\!\! 
 \mathrel{;}
\!\!\!\!\! 
 \begin{array}{c} \id \\ \oplus \\ \copygen \end{array} 
  \!\!\!\!\! 
 \mathrel{;}
\!\!\!\!\! 
 \begin{array}{c} \copyopgen \\ \oplus \\ \id \end{array} 
 \!\!\!\!\! 
 \mathrel{;}
\!\!\!\!\! 
 \begin{array}{c} \discardgen \\ \oplus \\ \id \end{array} 
\end{multline*}
The building blocks come from the signature of an algebraic theory---a \emph{symmetric monoidal theory} to be exact---$\ih_{k[s]}$ of interacting Hopf algebras~\cite{BSZ2,BSZ3,Za}. Terms as the one above can also be given an operational semantics~\cite{BSZ3} that accounts for the behavior of each basic component, and can thus be considered as the terms of a process algebra for signal flow graphs. The ideas of understanding complex dynamical systems by ``tearing'' them into more basic components, ``zooming'' to understand their individual behaviour and ``linking'' to obtain a composite system is at the core of the behavioural approach in control, originated by Willems~\cite{Wi3}. The algebra of symmetric monoidal categories thus seems a good fit for a formal account of these compositional principles.

This paper is the first to make this link explicit. Although~\cite{BSZ2,BSZ3,Za} made the connection between signal flow graphs and string diagrams, the operational semantics of~\cite{BSZ3,Za} demands that, initially, all the registers contain the value $0$, which is \emph{not} the standard systems theoretic semantics.  
Indeed, with this restriction, it is not difficult to see that the behaviours of~\eqref{eq:examplesfg} are those where the output stream is the same as the input stream. Operationally speaking, the input/output behaviour is thus that of a wire. Indeed, as shown in~\cite{BSZ3}, the operational and algebraic interpretations are closely related, and indeed, in $\ih_{k[s]}$ it is easy to prove that:
\begin{equation}\label{eq:exampleproof}
\lower12pt\hbox{$\includegraphics[height=1.2cm]{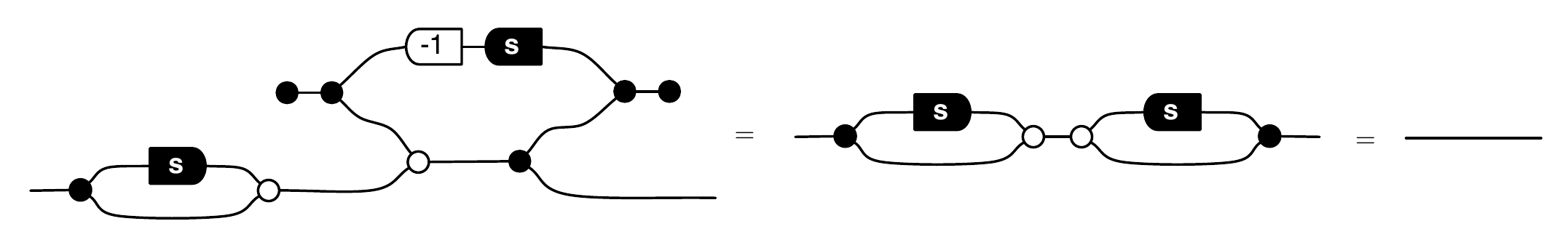}$} 
\end{equation}
The standard systems theoretic semantics allows registers to hold
\emph{arbitrary} values at the beginning of a computation. This extended notion
of behaviour is not merely a theoretical curiosity: it gives the class of
\emph{complete LTI discrete dynamical systems}~\cite{Wi3}, the class of
behaviours that is practically the lingua franca of control theory. The interest
of systems theorists is due to practical considerations: physical systems seldom
evolve from zero initial conditions.

For the sake of concreteness, consider the behaviour where the left register of~\eqref{eq:examplesfg} initially holds a $1$, the other a $2$. Now, inputting the stream $-1,1,-1,1,\dots$ results in the output $-2,2,-2,2,\dots$. This behaviour is clearly not permitted by a stateless wire, so the proof of~\eqref{eq:exampleproof} is \emph{not sound} for reasoning about circuits with this more liberal operational semantics. In this paper, we give a sound and complete theory to do just that.



Technically, the difference from previous work~\cite{BSZ1,BSZ3}---where streams
were handled with Laurent (formal power) series---is that here we must consider
\emph{biinfinite} streams: those sequences of elements of $\k$ that are infinite
in the past \emph{as well} as in the future, that is, elements of $\k^\z$.
Starting with the operational description, one obtains a biinfinite trajectory
by executing circuits forwards and backwards in time, for some initialisation of
the registers. The dynamical system defined by a circuit is the set of
trajectories obtained by considering all possible initialisations.

We obtain the equations by first showing that there is a full, but not faithful,
morphism from the prop $\cospan\mat\pk$ of cospans of matrices over the ring
$\pk$ to the prop $\ltids$ of complete LTI discrete dynamical systems. We rely
on~\cite{BSZ2,Za} for a presentation of $\cospan\mat\pk$. The result is a sound,
but not complete, proof system. The last ingredient is restricting our attention
from cospans to \emph{corelations}: this gives a faithful morphism, and allows
us to present the prop of corelations as a symmetric monoidal theory.


\smallskip 

The biinfinite trajectories admitted by~\eqref{eq:examplesfg} are a textbook
example of non-controllable behaviour, whereas in~\cite{BSZ1,BSZ3,Za,BE} all
definable behaviours were controllable.  We show that controllability has an
elegant structural characterisation in our setting.  Compositionality pays off
here, giving a new technique for reasoning about control of compound systems.
From the systems theoretic point of view, the results are promising since the
compositional, diagrammatic techniques we bring to the subject seem well-suited
to problems such as controllability of interconnections, of primary interest for
multiagent and spatially interconnected systems~\cite{OFM}.

\smallskip
Summing up, the original technical contributions of the paper are:
\begin{itemize}
\item a characterisation of the class of LTI systems as the category of corelations of matrices
\item a presentation of categories of corelations of matrices as symmetric monoidal theories
\item an operational semantics for the underlying syntax, which agrees with the standard systems theoretic
semantics of signal flow graph
\item a characterisation of controllability
\end{itemize}

\smallskip

Our work lies in the intersection of computer science, mathematics, and systems
theory.  From computer science, we use concepts  of formal semantics of
programming languages, with an emphasis on compositionality and a firm
denotational foundation for operational definitions.  From a mathematical
perspective, we obtain presentations of several relevant domains, and identify
the rich underlying algebraic structures.  For systems theory, our insight is
that mere matrices are not optimised for discussing behaviour; instead it is
profitable to use signal flow graphs, which treat linear subspaces rather than
linear transformations as the primitive concept and are thus closer to the idea
of system as a set of trajectories. At the core is the maxim---perhaps best
understood by computer scientists---that the right language allows deeper
insights into the underlying structure.

\paragraph{Structure of the paper.}
In \S\ref{sec.systems} we develop a categorical account of complete LTI discrete
dynamical systems, which serve as a denotational semantics for the graphical
language, introduced in \S\ref{sec.diagrams}, where we also derive the
equational characterisation. 
In \S\ref{sec.opsem} we discuss the operational semantics, and conclude  
in \S\ref{sec.control} with a structural account of controllability.

\section{Preliminaries}
We assume familiarity with basic concepts of linear algebra and category theory.
A \define{split mono} is a morphism $m\colon X\to Y$ such that there exists
$m'\colon Y\to X$ with $m'm=\idn_X$. 
For the remainder of the paper $\k$ is a field: for concreteness one may take
this to be the rationals $\mathbb{Q}$, the reals $\R$ or the booleans $\z_2$. 

A \define{prop} is a strict symmetric monoidal category where the set of objects
is the natural numbers $\nn$, and monoidal product ($\oplus$) on objects is
addition. Homomorphism of props are identity-on-objects strict symmetric
monoidal functors.

A \define{symmetric monoidal theory} (SMT) is a presentation of a prop: a pair
$(\Sigma,E)$ where $\Sigma$ is a set of \define{generators} $\sigma\colon m\to
n$, where $m$ is the \define{arity} and $n$ the \define{coarity}. A
$\Sigma$-term is a obtained from $\Sigma$, identity $\idn\colon 1\to 1$ and
symmetry $\tw\colon 2\to 2$ by composition and monoidal product, according to
the grammar
\[
  \tm\ ::=\ \sigma\ |\ \idn\ |\ \tw\ |\ \tm\mathrel{;}\tm\ |\ \tm\oplus \tm 
\]
where $\mathrel{;}$ and $\oplus$ satisfy the standard typing discipline that
keeps track of the domains (arities) and codomains (coarities)
\[
\frac{\tm: m\to d \quad \tm': d\to n}
{\tm\mathrel{;}\tm': m\to n}
\quad
\frac{\tm: m\to d \quad \tm': m'\to d'}
{\tm\oplus \tm': m+m'\to n+n'}
\]
The second component $E$ of an SMT is a set of \define{equations}, where an
equation is a pair $(\tm,\mu)$ of $\Sigma$-terms with compatible types, i.e.
$\tm,\mu\colon m\to n$ for some $m,n\in\mathbb{N}$.

Given an SMT $(\Sigma,E)$, the prop $\mathbf{S}_{(\Sigma,E)}$ has as arrows the
$\Sigma$-terms quotiented by the smallest congruence that includes the laws of
symmetric monoidal categories and the equations $E$. We will sometimes abuse
notation by referring to $\mathbf{S}_{(\Sigma,E)}$ as an SMT. Given an arbitrary prop
$\mathbb{X}$, a \define{presentation} of $\mathbb{X}$ is an SMT $(\Sigma,E)$
such that $\mathbb{X}\cong\mathbf{S}_{(\Sigma,E)}$.

String diagrams play an important role in this paper. Given a set of generators
$\Sigma$, we consider string diagrams to be arrows of
$\mathbf{S}_{(\Sigma,\varnothing)}$, that is, syntactic objects constructed by
composition from the generators and quotiented by the laws of symmetric monoidal
categories. For example, consider generators $2\to 1$ and $0\to 1$, which we
will draw as follows, with `dangling wires' on either side accounting for the
arity and coarity:
\[
\addgen \quad \zerogen
\]
Armed with the graphical notation, we can present sets of equations as equations
between diagrams. For example, the SMT of commutative monoids consists of the
generators above together with equations
\[
\includegraphics[height=.9cm]{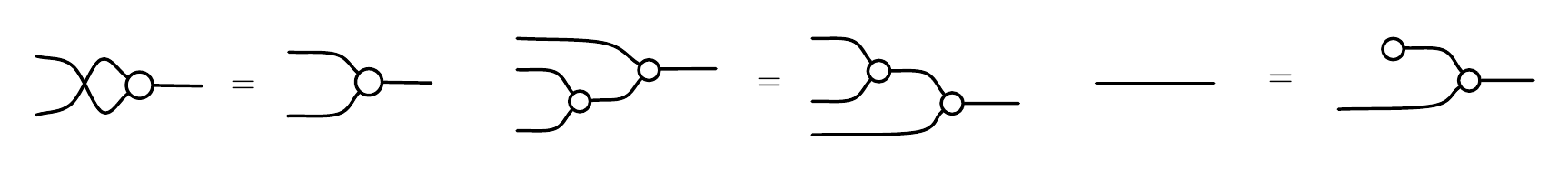}
\]
that respectively account for commutativity, associativity, and the unit law.

\section{Linear time-invariant dynamical systems} \label{sec.systems}

This section focusses on  
the mathematical domains of interest for the remainder of the paper. We rely on
fundamental definitions of the \emph{behavioural approach} in control, which is
informed by compositional considerations~\cite{Wi}. The concepts are standard in
systems theory, but our categorical insights are, to the best of our knowledge,
original.

Following Willems~\cite{Wi3}, a \define{dynamical system} $(T, W,\bb)$ is: a
\define{time axis} $T$, a \define{signal space} $W$, and a \define{behaviour}
$\bb \subseteq W^T$. We refer to $w \in \bb$ as \define{trajectories}. 
We consider discrete trajectories that are \define{biinfinite}: infinite in past
and future.  Our time axis is thus the integers $\z$. The signal space is
$\k^n$, where $n$ is the number of \define{terminals} of a system. These, in
engineering terms, are the interconnection variables that enable interaction
with an environment.

The dynamical systems of concern to us are thus specified by some natural number
$n$ and a subset $\bb$ of $(\k^n)^\z$. The sets $(\k^n)^\z$ are $\k$-vector
spaces, with pointwise addition and scalar multiplication; we restrict attention
to \emph{linear} systems, meaning that $\bb$ is a $\k$-linear subspace---i.e.
closed under pointwise addition and multiplication by $\k$-scalars---of
$(\k^n)^\z$. 

We partition terminals into a \emph{domain} and \emph{codomain}.  This may seem
artificial, in the sense that the assignment is arbitrary.  In particular, it is
crucial not to confuse the domains (codomains) with inputs (outputs). In spite
of the apparent contrivedness of choosing such a partition, Willems and others
have argued that it is vital for a sound theory of system \emph{decomposition};
indeed, it enables the ``tearing'' of Willems' tearing, zooming and linking
modelling methodology~\cite{Wi}.

Once the domains and codomains have been chosen, systems are linked by
connecting terminals. In models of physical systems this means variable coupling
or sharing~\cite{Wi}; in our discrete setting where behaviours are subsets of a
cartesian product---i.e.\ relations---it amounts to relational composition.
Since behaviours are both relations and linear subspaces, a central underlying
mathematical notion---as in previous work~\cite{BSZ1,BE}---is a linear relation. 
\begin{definition}
The monoidal category $\linrel_\k$ of $\k$-linear relations has $\k$-vector
spaces as objects, and as arrows from $V$ to $W$, linear subspaces of $V\oplus
W$, considered as $\k$-vector spaces.  Composition is relational: given $A\colon
U\to V$, $B\colon V\to W$, $A\mathrel{;}B\colon U\to W$ is the  relation
\[
  \left\{\,(u,w)\,\middle|\,\exists v\in V \textrm{ s.t. } \,(u,v)\in A \textrm{ and } (v,w)\in B\,\right\} 
\] 
that is easily checked to be linear. Finally, the monoidal product on both
objects and morphisms is direct sum.
\end{definition}

\smallskip
A behaviour is \define{time-invariant} when for every trajectory $w \in \bb$ and
fixed $i\in\z$, the trajectory $w'$ whose value at time $t\in\z$ is
\[
  w'(t) \Defeq w(t+i)
\]
is also in $\bb$.  Time-invariance brings with it a connection with the algebra
of polynomials.  Following the standard approach in control theory, going back
to Rosenbrock~\cite{Ro}, we work with polynomials over an indeterminate $s$ as
well as its formal inverse $s^{-1}$---i.e.\ the elements of the ring $\pk$.  The
introduction of the formal inverse $s^{-1}$ is a departure from previous
work~\cite{BSZ1,BSZ3} that dealt with Laurent streams (finite in the past,
infinite in the future); this is necessary as it does not make sense to define
the action of a polynomial fraction on a biinfinite stream. 

The indeterminate $s$ acts on a given biinfinite stream $w \in \k^\z$ as a
one-step delay, and $s^{-1}$ as its inverse, a one step acceleration: 
\[ 
  (s\cdot w) (t) \Defeq w(t-1),\quad (s^{-1}\cdot w)(t) \Defeq w(t+1).
\]
We can extend this, in the obvious linear, pointwise manner, to an action of any
polynomial $p\in \pk$ on $w$.  Since $\k^\z$ is a $\k$-vector space, any such
$p$ defines a $\k$-linear map $\k^\z\to \k^\z$, where $w \mapsto p\cdot w$.

Given this, we can view $n\times m$ matrices over $\pk$ as $\k$-linear maps from
$(\k^m)^\z$ to $(\k^n)^\z$. This can be explained succinctly as a functor from
the prop $\mat\pk$, defined below, to the category of $\k$-vector spaces and
linear transformations $\vect_\k$.

\begin{definition}
  The prop $\mat\pk$ has as arrows $m \to n$ the $n\times m$-matrices over
  $\pk$. Composition is matrix multiplication, and the monoidal product of $A$
  and $B$ is $\left[\begin{smallmatrix} A & 0 \\ 0 & B
  \end{smallmatrix}\right]$. The symmetries are permutation matrices.
\end{definition}

The functor of interest
\[
  \vectfun\maps \mat\pk \longrightarrow \vect_\k
\]
takes a natural number $n$ to the vector space $(\k^n)^\z$, and an $n\times m$
matrix to the induced linear transformation $(\k^m)^\z \to (\k^n)^\z$. Note that
$\vectfun$ is faithful.

\smallskip
The final restriction on the set of behaviours is called \emph{completeness},
and is a touch more involved. For $t_0,t_1 \in \z$, $t_0 \le t_1$, write
$w|_{[t_0,t_1]}$ for the restriction of $w: \z \to \k^n$ to the set $[t_0,t_1] =
\{t_0, t_0+1, \dots, t_1\}$. Write  $\bb|_{[t_0,t_1]}$ for the set of the
restrictions of all trajectories $w \in \bb$ to $[t_0,t_1]$.  Then $\bb$ is
\define{complete} when $w|_{[t_0,t_1]} \in \bb|_{[t_0,t_1]}$ for all $t_0,t_1
\in \z$ implies $w \in \bb$. 


\begin{definition}
  A \define{linear time-invariant (LTI) behaviour} comprises a domain
  $(\k^m)^\z$, a codomain $(\k^n)^\z$, and a subset $\bb \subseteq (\k^m)^\z
  \oplus (\k^n)^\z$ such that $(\z,\k^m \oplus \k^n,\bb)$ is a complete, linear,
  time-invariant dynamical system.
\end{definition}

The algebra of LTI behaviours is captured concisely as a prop.
\begin{proposition} \label{prop.ltidsiswelldefined}
  There exists a prop $\ltids$ 
  with morphisms $m \to n$ the LTI behaviours with domain $(\k^m)^\z$ and
  codomain $(\k^n)^\z$. Composition is relational. The monoidal product is
  direct sum.
\end{proposition}


The proof of Proposition~\ref{prop.ltidsiswelldefined} relies on \emph{kernel
representations} of LTI  systems.
The following result 
 lets us pass between
behaviours and polynomial matrix algebra.
%
%
\begin{theorem}[Willems {\cite[Th. 5]{Wi3}}] \label{thm.kernelreps}
  Let $\bb$ be a subset of $(\k^n)^\z$ for some $n \in \mathbb N$. Then $\bb$ is
  an LTI behaviour iff there exists $M \in
  \mat\pk$ such that $\bb = \mathrm{ker}(\vectfun M)$.
\end{theorem}


The prop $\mat \pk$ is equivalent to the category $\fmod\pk$ of
finite-dimensional free $\pk$-modules. Since $\fmod R$ over a principal ideal
domain (PID) $R$ has finite colimits \cite{BSZ2}, and $\pk$ is a PID, $\mat \pk$
has finite colimits, and pushouts in particular.

We can thus define the prop $\cospan\mat\pk$ where arrows are equivalence
classes of cospans of matrices. Up to isomorphism, this means we may consider
arrows $m \to n$ in this prop to comprise a natural number
$d$ together with a $d\times m$ matrix $A$ and a $d\times n$ matrix $B$; we
write this $m\xrightarrow{A} d \xleftarrow{B}n$. Composition is given by
pushout; see Appendix~\ref{app.cospans} for further details.

%
We can then extend $\vectfun$ to the functor
\[
  \cospanfun \maps \cospan\mat\pk \longrightarrow \linrel_\k
\]
where on objects $\cospanfun(n)=\vectfun(n)=(\k^n)^\z$, and on arrows
\[
  m\xrightarrow{A} d \xleftarrow{B}n
\]
maps to
\begin{equation}\label{eq.K}
  \big\{(\mathbf{x},\mathbf{y}) \,\big|\,
  (\vectfun A)\mathbf{x} = (\vectfun B)\mathbf{y}\big\} 
  \subseteq (\k^m)^\z \oplus (\k^n)^\z.
\end{equation}
It is straightforward to prove that this is well-defined.
\begin{proposition}\label{prop.funct}
$\cospanfun$ is a functor.
\end{proposition}
\begin{proof}
Identities are clearly preserved; it suffices to show that composition is too.
Consider the diagram below, where the pushout is in $\mat\pk$.
\[
\xymatrix{
  {}\ar[dr]_{A_1} & & \ar[dl]^{B_1}  
  \ar[dr]_{A_2} & & \ar[dl]^{B_2} {} 
 \\
& \ar[dr]_{C} & & \ar[dl]^{D}  \\
& & {\save*!<0cm,-.5cm>[dl]@^{|-}\restore}
}
\]
To show that $\cospanfun$ preserves composition we must verify that 
\begin{multline*}
\{(\mathbf{x},\mathbf{y})\,|\, \theta CA_1\mathbf{x} = \theta DB_2\mathbf{y}\}
= \\
\{(\mathbf{x},\mathbf{z})\,|\, \theta A_1\mathbf{x} = \theta B_1\mathbf{z}\} ; 
\{(\mathbf{z},\mathbf{y})\,|\, \theta A_2\mathbf{z} = \theta B_2\mathbf{y}\}.
\end{multline*}
The inclusion $\subseteq$ follows from properties of pushouts in $\mat\pk$ (see
\cite[Prop. 5.7]{BSZ2}).  To see $\supseteq$, we need to show that if there
exists $\mathbf{z}$ such that $\theta A_1\mathbf{x} = \theta B_1\mathbf{z}$ and
$\theta A_2\mathbf{z} = \theta B_2\mathbf{y}$, then $\theta CA_1\mathbf{x} =
\theta DB_2\mathbf{y}$.  But $\theta CA_1\mathbf{x}=\theta CB_1\mathbf{z}=\theta
DA_2\mathbf{z}=\theta DB_2\mathbf{y}$.
\end{proof}


Rephrasing the definition of $\cospanfun$ on morphisms~\eqref{eq.K}, the
behaviour consists of those $(\mathbf{x},\mathbf{y})$ that satisfy
\[
  \vectfun\left[\begin{array}{cc} A & -B\end{array}\right]
  \left[\begin{array}{c}\mathbf{x}\\\mathbf{y}\end{array}\right] 
  = \mathbf{0},
\]
so one may say---ignoring for a moment the terminal domain/codomain assignment---that 
\[ 
  \cospanfun (\xrightarrow{A}\xleftarrow{B}) = \ker \vectfun\left[
  \begin{array}{cc} A & -B\end{array}\right].
\]

With this observation, as a consequence of Theorem~\ref{thm.kernelreps},
$\cospanfun$ has as its image (essentially) the prop $\ltids$. This proves
Proposition \ref{prop.ltidsiswelldefined}. We denote the corestriction of
$\cospanfun$ by $\cospanfunrest$. We thus have a full functor:
\[
  \cospanfunrest \maps \cospan\mat\pk \longrightarrow \ltids.
\]

\begin{remark}\label{rmk:faithfulness}
It is important for the sequel to note that $\cospanfunrest$ is \emph{not} faithful.
For instance, $\cospanfunrest(\xrightarrow{[1]}\xleftarrow{[1]}) =
\cospanfunrest(\xrightarrow{ \left[\begin{smallmatrix} 1\\ 0\end{smallmatrix}\right]
} \xleftarrow{ \left[\begin{smallmatrix} 1\\ 0\end{smallmatrix}\right]} )$, yet
the cospans are not equivalent.
\end{remark}


\section{Presentation of $\ltids$} \label{sec.diagrams}
In this section we give a presentation of $\ltids$ as an SMT. This means both
that \emph{(i)} we obtain a syntax---conveniently expressed using string
diagrams---for specifying every LTI behaviour, and \emph{(ii)} a sound and
fully complete equational theory for reasoning about them.

\subsection{Syntax}
We start by describing the graphical syntax of dynamical systems, the arrows of
the category $\syntax = \mathbf{S}_{(\Sigma,\varnothing)}$, where $\Sigma$ is
the set of generators:
\begin{multline}\label{eq:generators}
\{
\addgen,
\zerogen,
\copygen,
\discardgen,
\delaygen, \\
\addopgen,
\zeroopgen,
\copyopgen,
\discardopgen,
\delayopgen
\} \\
\cup \{ \scalargen \mid a\in\k \,\} \cup \{ \scalaropgen \mid a\in\k \,\}
\end{multline}
For each generator, we give its \define{denotational semantics}, an LTI
behaviour, thereby defining a prop morphism $\llbracket - \rrbracket:
\syntax \to\ltids$.
\[
\addgen \!\!\mapsto \{\,( \left( 
  {\begin{smallmatrix}\tau \\ \upsilon\end{smallmatrix}} \right),\, \tau+\upsilon) \mid \tau,\upsilon \in \k^\z \,\}\maps 2\to 1
\]
\[
\quad
\zerogen \!\!\mapsto
\{\,((),0)\,\} \subseteq \k^\z\maps 0\to 1
\]
\[
\copygen \!\!\mapsto 
\{\, (\tau, \left( \begin{smallmatrix} \tau \\ \tau\end{smallmatrix} \right)) \mid \tau\in \k^\z \,\}\maps1\to 2
\]
\[
\quad
\discardgen \!\!\mapsto
\{\,(\tau,()) \mid \tau \in \k^\z\,\} \maps 1\to 0
\]
\[
\scalargen  \!\!\mapsto
\{\, (\tau, a\cdot\tau) \mid \tau\in\k^\z \, \} \maps 1 \to 1 \quad (a\in\k)
\]
\[
\ 
\delaygen \!\!\mapsto
\{\, (\tau, s\cdot\tau) \mid \tau\in\k^\z\,\} \maps 1 \to 1
\]

\noindent 
The denotations of the mirror image generators are the opposite relations.
Parenthetically, we note that a finite set of generators is possible if working
over a finite field, or the field $\mathbb{Q}$ of rationals, \emph{cf.}\
\S\ref{subsec:eqhopf}.

The following result guarantees that the syntax is fit for purpose: every
behaviour in $\ltids$ has a syntactic representation in $\syntax$.
\begin{proposition}\label{prop.syntaxfull}
  $\llbracket - \rrbracket: \syntax\to\ltids$ is full.
\end{proposition}
\begin{proof}
  The fact that $\llbracket - \rrbracket$ is a prop morphism is immediate since
  $\mathbb{S}$ is free on the SMT $(\Sigma,\varnothing)$ and there are no
  equations.  Fullness is a consequence of the fact that $\llbracket - \rrbracket$
  factors as the composite of two full functors as follows:
  \[
    \xymatrixrowsep{2pc}
    \xymatrix{
      \syntax \ar[d] \ar[dr]^{\llbracket-\rrbracket} \\
      \cospan\mat\pk \ar[r]_-{\cospanfunrest} & \ltids
    }
  \]
  The functor $\cospanfunrest$ is full by definition. The existence and fullness
  of the functor $\mathbb{S}\to\cospan\pk$ follows from \cite[Th. 3.41]{Za}. We
  give details in the next two subsections.
\end{proof}

Having defined the syntactic prop $\syntax$ to represent arrows in $\ltids$, our
task for this section is to identify an equational theory that
\emph{characterises} $\ltids$: i.e. one that is sound and fully complete. The
first step is to use the existence of $\cospanfun$: with the results
of~\cite{BSZ2,Za} we can obtain a presentation for $\cospan\mat\pk$. This is
explained in the next two subsections: in \S\ref{subsec:eqhopf} we present
$\mat\pk$ and in \S\ref{subsec:eqcospan}, the equations for cospans of matrices.


\subsection{Presentation of $\mat\pk$}\label{subsec:eqhopf}


To obtain a presentation of $\mat\pk$ as an SMT we only require 
some of the generators:
\begin{multline*}
\{
\addgen,
\zerogen,
\copygen,
\discardgen,
\delaygen,
\delayopgen\} \\
\cup \{ \scalargen \mid a\in\k \,\}
\end{multline*}
and the
following equations. 
First, the white and the black structure forms a (bicommutative) bimonoid:
\[
\includegraphics[width=.4\textwidth]{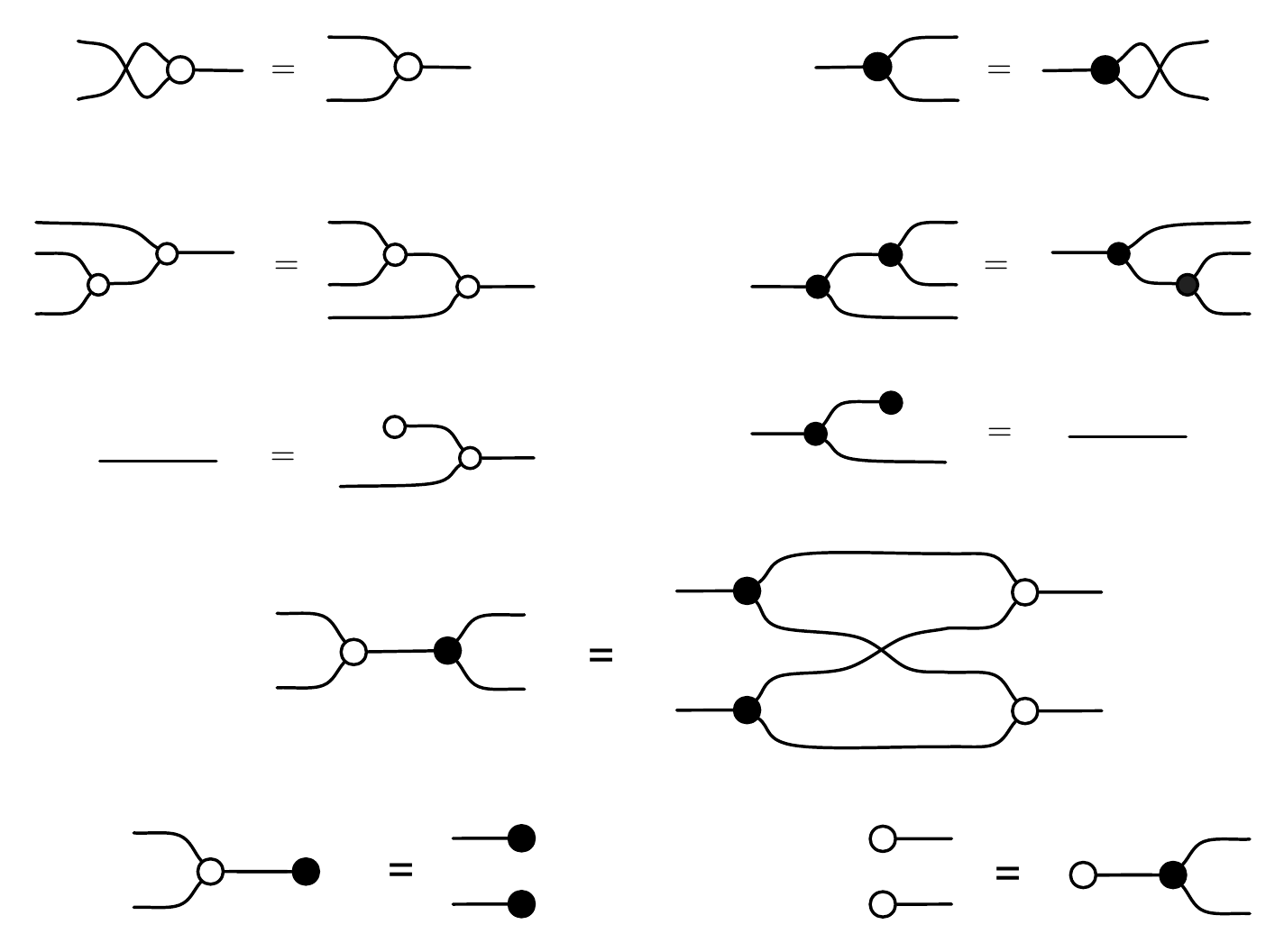}
\]
Next, the formal indeterminate $s$ is compatible with the bimonoid structure
and has its mirror image as a formal inverse.
\[
\includegraphics[height=3.1cm]{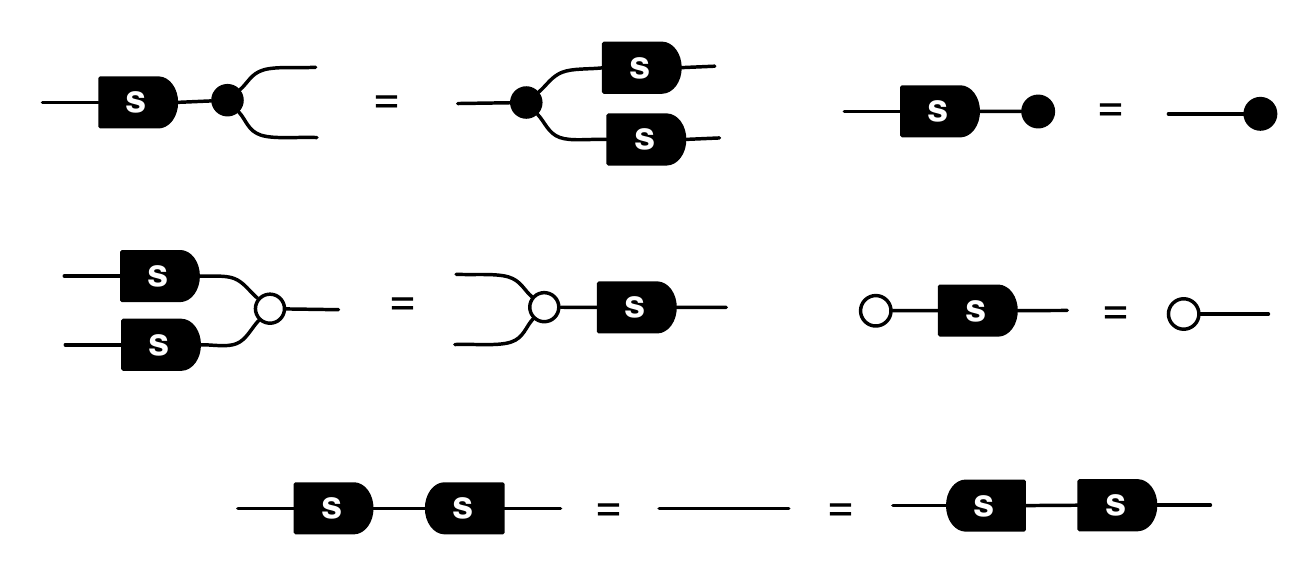}
\]
Finally, we insist that the algebra of $\k$ be compatible with the bimonoid structure and commute with $s$.
The equations below, in which $a,b\in\k$, are redundant if $\k=\mathbb{Q}$; instead,
one needs an additional generator, the antipode, identified with the scalar
$-1$. The details are not important; we merely mention that henceforward  
\lower6pt\hbox{$\includegraphics[height=0.6cm]{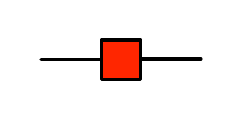}$}
is synonymous with
\lower6pt\hbox{$\includegraphics[height=0.6cm]{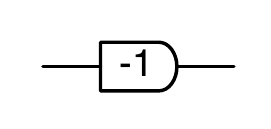}$}. 
\[
\includegraphics[height=5.3cm]{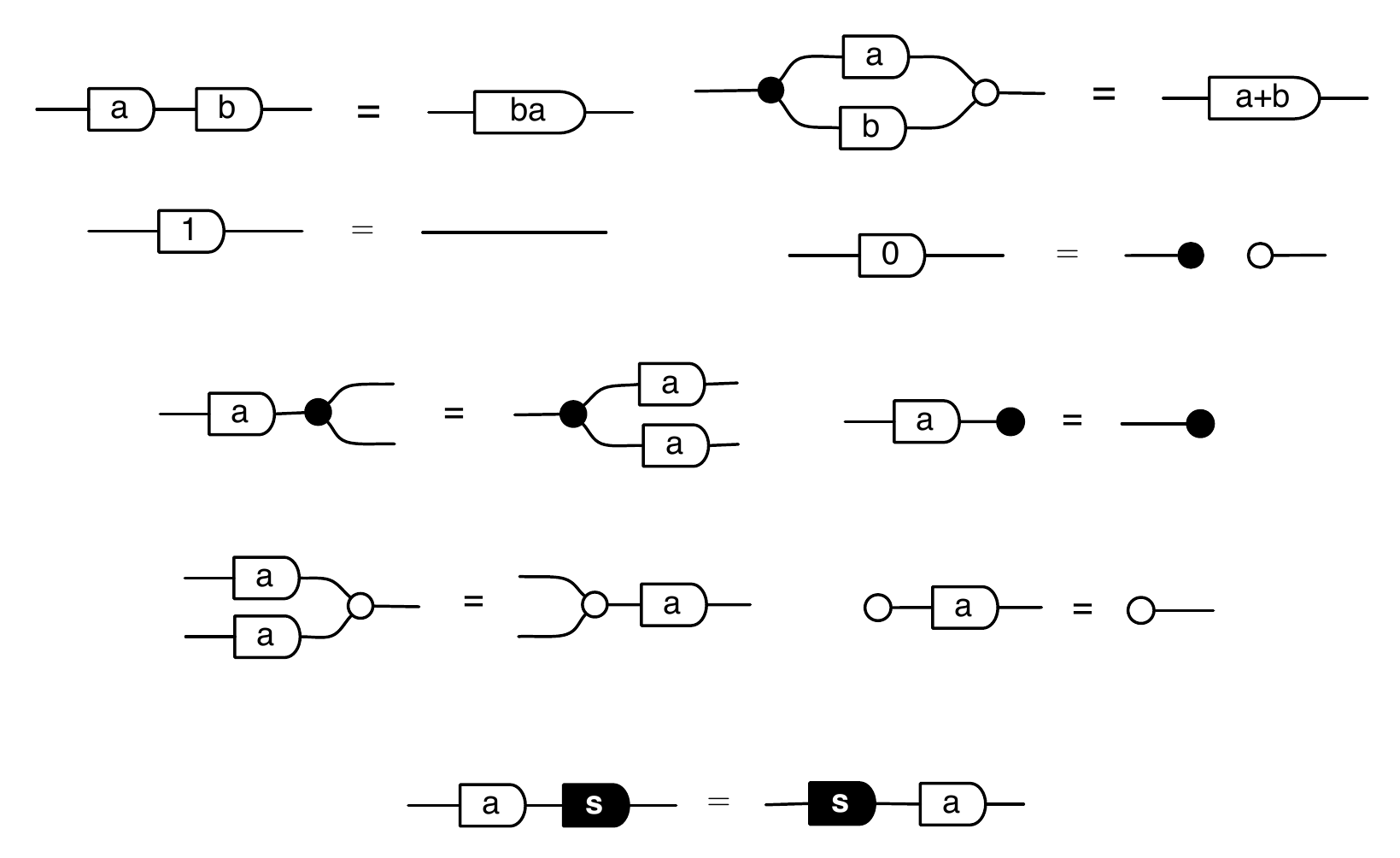}
\]
Let $\ha_\pk$ be the prop induced by the SMT consisting of the equations
above. The following follows from~\cite[Prop.~3.9]{Za}.
\begin{proposition}
$\mat\pk\cong\ha_\pk$.
\end{proposition}

Arrows of $\mat\pk$ are matrices with polynomial entries, but it may not
be immediately obvious to the reader to the reader how polynomials arise with
the string diagrammatic syntax. We illustrate this in the following remark.

\begin{remark}\label{rem:polys}
Any polynomial $p=\sum_{i=u}^v a_i s^i$, where $u\leq v\in\z$ and with
coefficients $a_i\in\k$, can be written graphically using the building blocks of
$\mathbb{HA}_\pk$. Rather than giving a tedious formal construction, we
illustrate this with an example for $\k=\R$. A term for $3s^{-3}-\pi
s^{-1}+s^{2}$ is:
\[
\includegraphics[height=1.7cm]{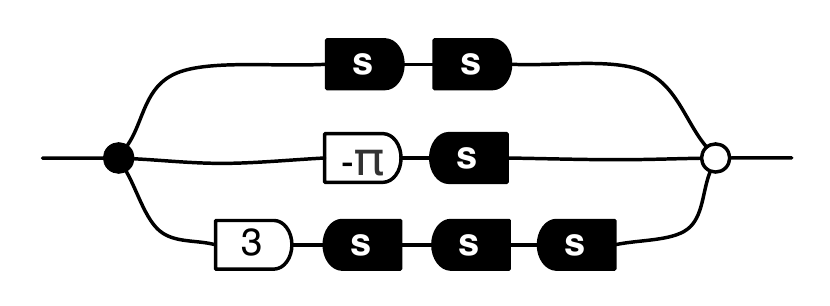}
\]

As an arrow $1 \to 1$ in $\mat\pr$, the above term represents a $1\times
1$-matrix over $\pr$. To demonstrate how higher-dimensional matrices can be
written, we also give a term for the $2 \times 2$-matrix $\begin{bmatrix} 2 & 3s
  \\ s^{-1} & s+1 \end{bmatrix}$:
\[
\includegraphics[height=3cm]{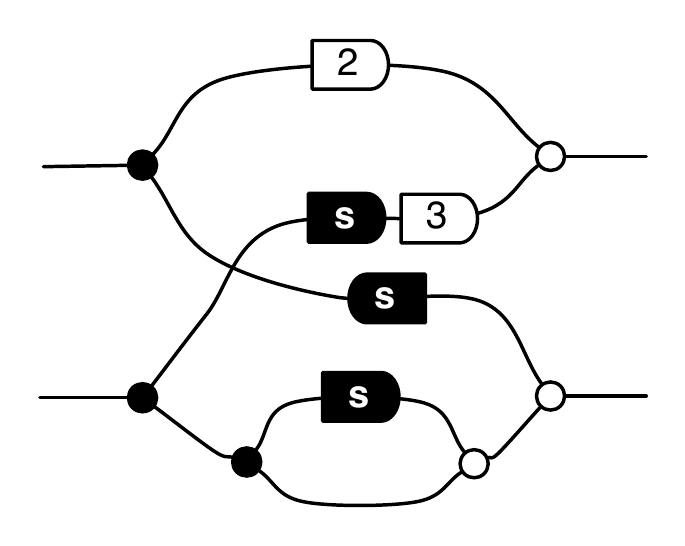}
\]
The above examples are intended to be suggestive of a normal form for terms in
$\ha_\pk$; for further details see \cite{Za}.
\end{remark}


\subsection{Presentation of $\cospan\mat\pk$\label{subsec:eqcospan}}


To obtain the equational theory of $\cospan\mat \pk$ we need the full set of
generators~\eqref{eq:generators},
along with the equations of $\ha_\pk$, their mirror images, and the
following
\[
\includegraphics[height=4cm]{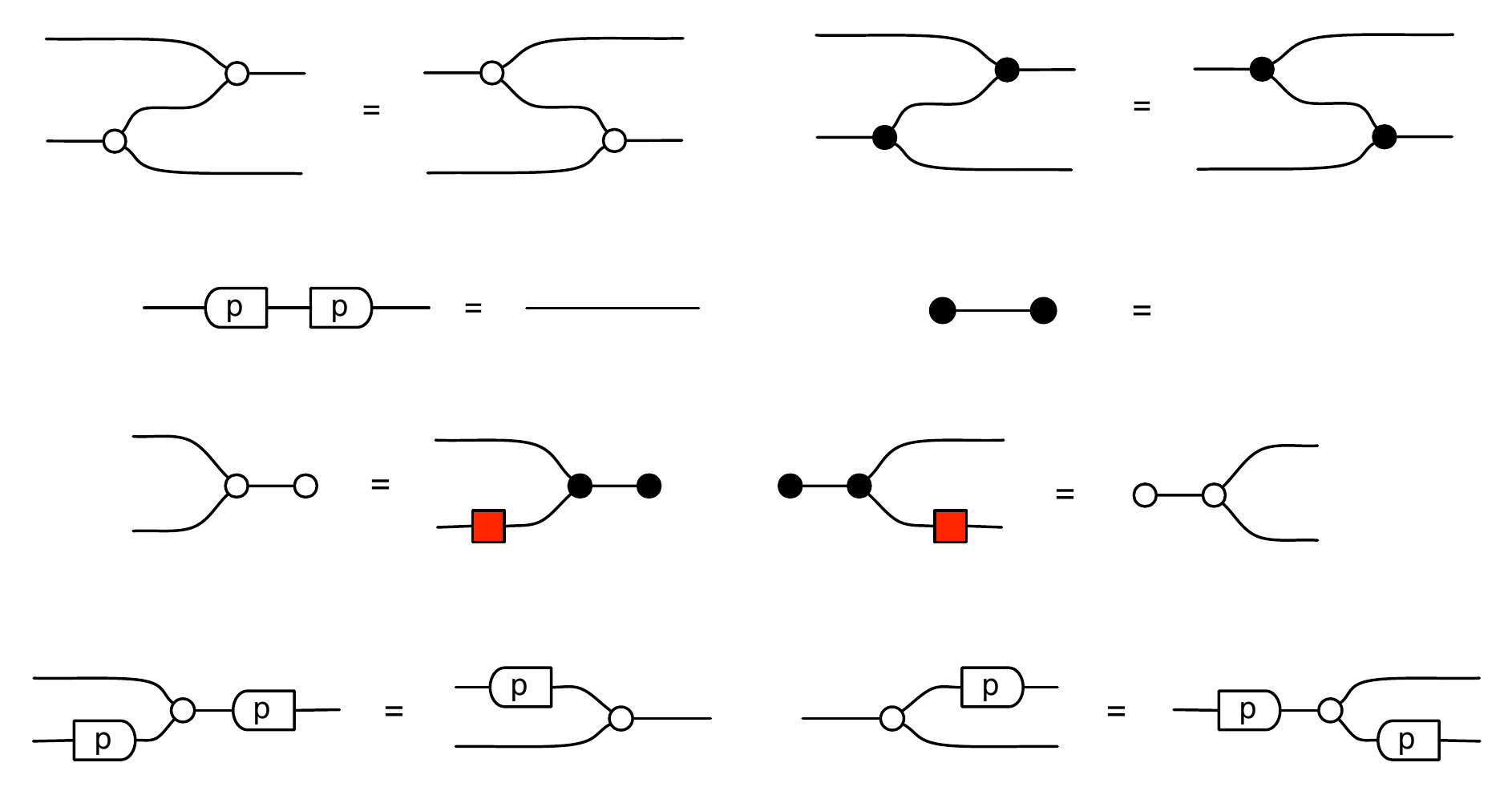}
\]
where $p$ ranges over the nonzero elements of $\pk$ (see
Remark~\ref{rem:polys}). The equations of $\ha_\pk$ ensure the generators of
$\ha_\pk$ behave as morphisms in $\mat\pk$, while their mirror images ensure the
remaining generators behave as morphisms in the opposite category
$\mat\pk^{\mathrm{op}}$. The addition equations above govern the interaction
between these two sets of generators, axiomatising pushouts in $\mat\pk$.

Let $\ihcsp$ denote the resulting SMT.
The procedure for obtaining the equations from a distributive law of props 
is explained in~\cite[\S{3.3}]{Za}.
\begin{proposition}[Zanasi~{\cite[Th.~3.41]{Za}}]\label{prop:cospanpresentation}
$\cospan \mat\pk \cong \ihcsp$.
\end{proposition}

Using Prop.~\ref{prop:cospanpresentation} and the existence of $\cospanfunrest$,
the equational theory of $\ihcsp$ is a sound proof system for reasoning about
$\ltids$. Due to the fact that $\cospanfunrest$ is not faithful (see
Remark~\ref{rmk:faithfulness}), however, the system is not complete. Achieving
completeness is our task for the remainder of this section.


\subsection{Corelations}



In the category of sets, relations can be identified with jointly-mono spans of
functions; that is, those spans $X\xleftarrow{p}R\xrightarrow{q}Y$ where the
induced map $R\xrightarrow{\langle p,q\rangle}X\times Y$ is injective.
Corelations are a dual concept: we consider cospans
$X\xrightarrow{i}S\xleftarrow{j}Y$ where the induced map
$X+Y\xrightarrow{[i,j]}S$ is surjective. To make sense of this more generally,
one needs a category with a factorisation system. (See
Appendix~\ref{app.factorisation} for a definition.)  In this subsection we
identify a factorisation system in $\mat\pk$, and show that the induced prop
$\corel\mat\pk$ of corelations is isomorphic to $\ltids$.  We then give a
presentation of $\corel\mat\pk$ and arrive at a sound and fully complete
equational theory for $\ltids$.


\begin{proposition}\label{prop.matfactorisation}
  Every morphism $R \in \mat\pk$ can be factored as $R = BA$, where $A$ is an
  epi and $B$ is a split mono.
\end{proposition}
\begin{proof}
We use the Smith normal form~\cite[Section~6.3]{Ka}. In 
Appendix~\ref{app.prop.matfactorisation}.
\end{proof}
A more careful examination of Proposition~\ref{prop.matfactorisation} yields
the following.

\begin{corollary} \label{cor.episplitmono}
  Let $\mathcal E$ be the subcategory of epis and
  $\mathcal M$ the subcategory of split monos. The pair $(\mathcal
  E,\mathcal M)$ is a factorisation system on $\mat\pk$, with $\mathcal M$
  stable under pushout.
\end{corollary}
\begin{proof}
In Appendix~\ref{app.cor.episplitmono}. See Appendix~\ref{app.factorisation} for background.
\end{proof}

Given finite colimits and an $(\mathcal E,\mathcal M)$-factorisation system with
$\mathcal M$ stable under pushouts, we may define a category of corelations. The
morphisms are isomorphism classes of cospans $X \xrightarrow{i} S \xleftarrow{j}
Y$ where the copairing $[i,j]\maps X+Y \to S$ is in $\mathcal E$.  Composition
is given by pushout, as in the category of cospans, followed by factorising the
copairing of the resulting cospan. This is a dualisation of a well-known
construction of relations from spans~\cite{JW}; the details can be found
in~\cite{Fo}.

\begin{definition}
  The prop $\corel \mat \pk$ has as morphisms equivalence classes of
  jointly-epic cospans in $\mat\pk$.  
\end{definition}

We have a full morphism
\[
  F\maps \cospan \mat \pk \longrightarrow \corel \mat \pk
\]
mapping a cospan to its jointly-epic counterpart given by the
 factorisation system. 
Then $\cospanfunrest$
factors through $F$, giving a commutative triangle:
\[
  \xymatrixrowsep{2pc}
  \xymatrix{
    \cospan\mat\pk \ar[d]_F \ar[dr]^{\cospanfunrest} \\
    \corel\mat\pk \ar[r]_-{\Phi} & \ltids
  }
\]
The morphism $\Phi$ along the base of this triangle is an isomorphism of props, 
and this is our main technical result, Theorem~\ref{thm.main}. The proof relies
on the following celebrated, beautiful result of systems theory.

\begin{proposition}[Willems {\cite[p.565]{Wi3}}] \label{prop.magic}
  Let $M,N$ be matrices over $\pk$. Then $\ker \vectfun M \subseteq \ker \vectfun N$ iff there exists a matrix $X$ such that $XM = N$.
\end{proposition}

Further details and a brief history of the above proposition can be found in
Schumacher \cite[pp.7--9]{Sc}. 

\begin{theorem}\label{thm.main}
  There is an isomorphism of props 
  \[
    \Phi\maps \corel\mat\pk \longrightarrow \ltids
  \]
  taking a corelation $\xrightarrow{A}\xleftarrow{B}$ 
  to $\cospanfunrest(\xrightarrow{A}\xleftarrow{B}) = {\ker\theta [A \ -B]}$.
\end{theorem}
\begin{proof}
  To prove functoriality, start from the functor $\vectfun\maps\mat\pk
  \linebreak \to
  \vect_\k$. Now (i) $\vect_k$ has an epi-mono factorisation system, (ii)
  $\vectfun$ maps epis to epis and (iii) split monos to monos, so $\vectfun$
  preserves factorisations. Since it is a corollary of Prop. \ref{prop.funct}
  that $\theta$ preserves colimits, it follows that $\vectfun$ extends to a prop
  functor $\Psi\maps\corel\mat\pk \to \corel\vect_\k$. But $\corel\vect_\k$ is
  isomorphic to $\linrel_\k$ (see Appendix \ref{app.corellinrel} and~\cite{Fo}).
  By Theorem \ref{thm.kernelreps}, the image of $\Psi$ is $\ltids$, and taking
  the corestriction to gives us precisely $\Phi$, which is therefore a full
  morphism of props.

  As corelations $n \to m$ are in one-to-one correspondence with epis out of
  $n+m$, to prove faithfulness it suffices to prove that if two epis $R$ and $S$
  with the same domain have the same kernel, then there exists an invertible
  matrix $U$ such that $UR =S$. This is immediate from
  Proposition~\ref{prop.magic}: if $\ker R= \ker S$, then we can find $U, V$
  such that $UR = S$ and $VS = R$ Since $R$ is an epimorphism, and since $VUR =
  VS = R$, we have that $VU=1$ and similarly $UV =1$. This proves that any two
  corelations with the same image are isomorphic, and so $\Phi$ is full and
  faithful.  
\end{proof}


\subsection{Presentation of $\corel\mat\pk$}


As a consequence of Theorem~\ref{thm.main}, the task of obtaining a presentation of $\ltids$
is that of obtaining one for $\corel\mat\pk$. 
%
To do this, we start with the presentation
$\mathbb{IH}^{\textsf{Csp}}$ for $\cospan\mat\pk$ of \S\ref{subsec:eqcospan}; the task of
this section is to identify the additional equations that equate
exactly those cospans that map via $F$ to the same corelation.


In fact, only one extra equation is required, the ``white bone law'':
\begin{equation}\label{eq.whitebone}
  \lower8pt\hbox{$\includegraphics[height=.7cm]{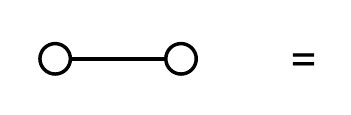}$}\qquad\qquad
\end{equation}
Expressed in terms of cospans, \eqref{eq.whitebone} asserts that  
$0\rightarrow 1\leftarrow 0$ and $0\rightarrow 0\leftarrow 0$ are identified: indeed, 
the two clearly yield the same corelation.
Let $\mathbb{IH}^{\mathsf{Cor}}$ be the SMT obtained
from the equations of $\mathbb{IH}^{\mathsf{Csp}}$ and equation~\eqref{eq.whitebone}.
\begin{theorem}
$\corel \mat \pk \cong \mathbb{IH}^{\mathsf{Cor}}$.
\end{theorem}
\begin{proof}
Since equation~\eqref{eq.whitebone} holds in $\corel\mat\pk$, we have
a full morphism $\ihcor \to \corel\mat\pk$; it remains
to show that it is faithful. It clearly suffices to show that in the equational
theory $\ihcor$ one can prove that every cospan is equal
to its corelation. Suppose then that $m \xrightarrow{A} k \xleftarrow{B} n$ is a cospan
and $m \xrightarrow{A'} k' \xleftarrow{B'} n$ its corelation. Then, by definition, 
there exists a split mono $M\maps k'\to k$ such that $MA'=A$ and $MB'=B$. Moreover, by the construction
of the epi-split mono factorisation in $\mat\pk$, $M$ is of the form $\lower10pt\hbox{$\includegraphics[height=1cm]{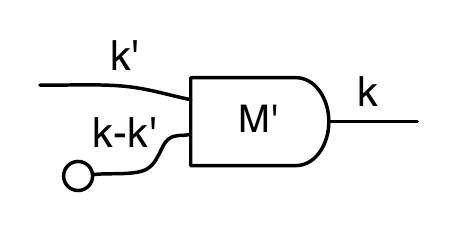}$}$ where $M'\maps k\to k$ is invertible. We can now give the derivation
in $\mathbb{IH}^{\mathsf{Cor}}$:
\begin{align*}
\lower7pt\hbox{$\includegraphics[height=.7cm]{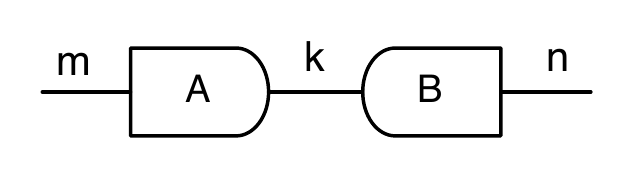}$} \quad
&\stackrel{\ihcsp}{=} \quad
\lower10pt\hbox{$\includegraphics[height=1cm]{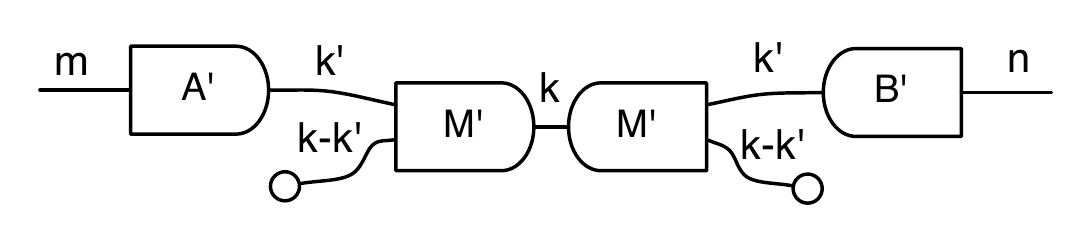}$} \\
\quad &\stackrel{\ihcsp}{=} \qquad 
\lower8pt\hbox{$\includegraphics[height=.8cm]{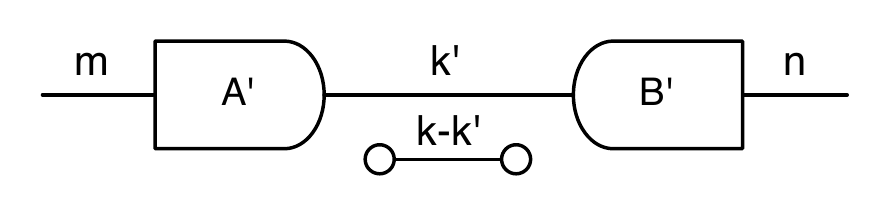}$} \\
&\stackrel{\eqref{eq.whitebone}}{=} \qquad\qquad
\lower7pt\hbox{$\includegraphics[height=.7cm]{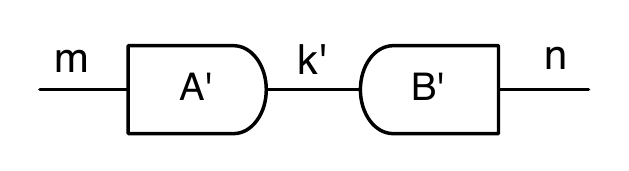}$}
\end{align*}
\end{proof}

%

We therefore have a sound and fully complete equational theory for LTI systems,
and also a normal form for each LTI system: every such system can be written, in
an essentially unique way, as a jointly-epic cospan of terms in $\ha_\pk$ in
normal form.

\begin{remark}
 $\ihcor$ (summarised in Appendix~\ref{app.equations}) can also be described as having the equations of $\ih_\pk$~\cite{BSZ2,Za}, but
\emph{without} $\!\!\lower4pt\hbox{$\includegraphics[height=.5cm]{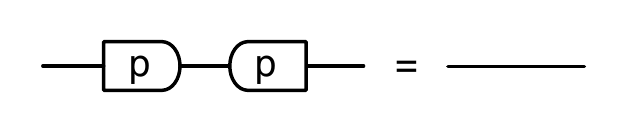}$}$,
and the related $\!\!\lower7pt\hbox{$\includegraphics[height=.8cm]{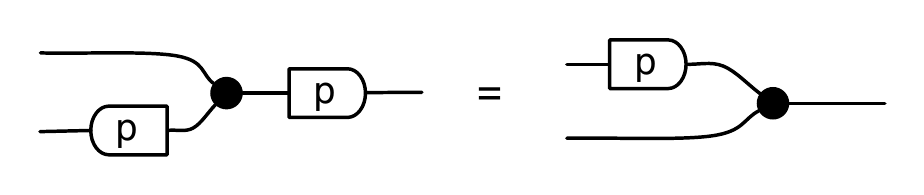}$}$
and $\!\!\lower7pt\hbox{$\includegraphics[height=.8cm]{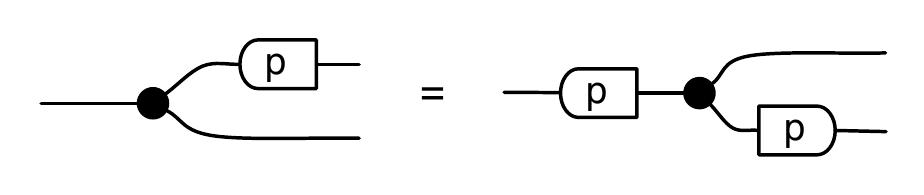}$}$.
Our results generalise; given any PID $R$ we have (informally speaking):
\begin{align*}
  \ihcor_R &= \ih_R -
  \left\{
    \begin{array}{c}
      \lower6pt\hbox{$\includegraphics[height=.6cm]{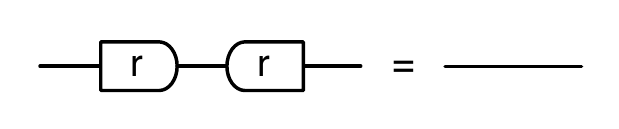}$}, \\
      \lower7pt\hbox{$\includegraphics[height=.7cm]{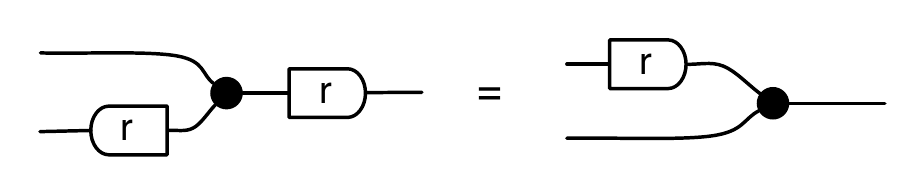}$},\\
      \lower7pt\hbox{$\includegraphics[height=.7cm]{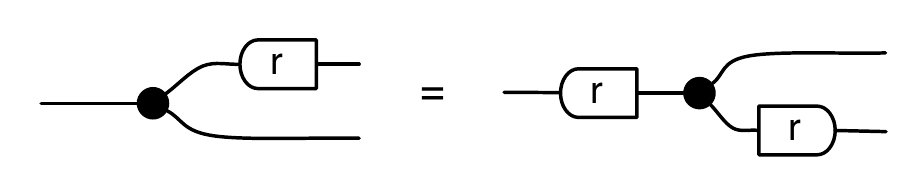}$}
    \end{array}
    \middle|\,r\neq 0\in R\,
  \right\} \\
  &\cong \corel \mat R
\end{align*}
and, because of the transpose duality of matrices:
\begin{align*}
  \ih^{\mathsf{Rel}}_R &= \ih_R - 
  \left\{
    \begin{array}{c} 
      \lower6pt\hbox{$\includegraphics[height=.6cm]{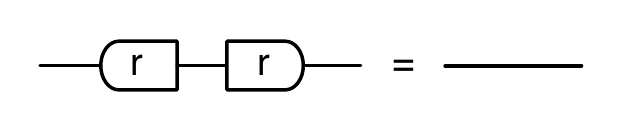}$},\\
      \lower7pt\hbox{$\includegraphics[height=.7cm]{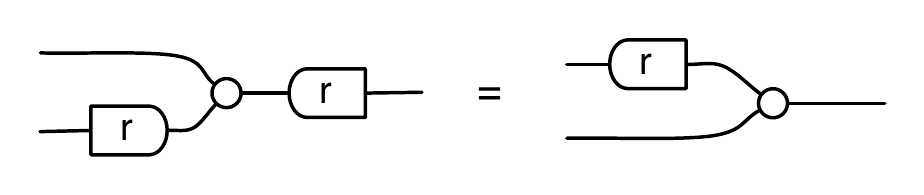}$}, \\
      \lower7pt\hbox{$\includegraphics[height=.7cm]{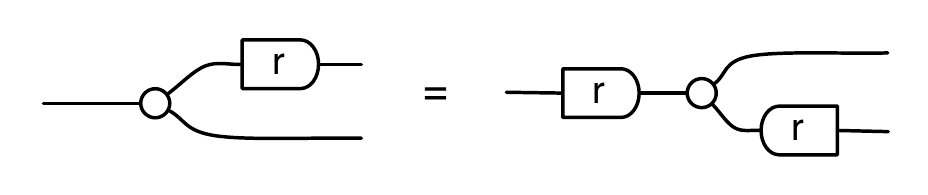}$}
    \end{array}
    \middle|\,r\neq 0\in R\,
  \right\}\\
  &\cong \mathsf{Rel} \mat R.
\end{align*}
\end{remark}

The similarity between our equational presentation of $\ihcor_{\k[s,s^{-1}]}$ and
the equational presentation of $\ih_{\k[s]}$  given in~\cite{BSZ1,BSZ3} is remarkable,
considering the differences between the intended semantics of signal flow graphs
that string diagrams in those theories represent, as well as the underlying mathematics
of streams, which for us are elements of the $\k$ vector space $\k^{\mathbb{Z}}$ and
in~\cite{BSZ1,BSZ3} are Laurent series. We contend that this is evidence of the \emph{robustness}
of the algebraic approach: the equational account of how the various components of signal
flow graphs interact is, in a sense, a higher-level specification than the technical details of the underlying
mathematical formalisation.

\section{Operational semantics} \label{sec.opsem}

In this section we relate the denotational account given in previous sections
with an operational view. 

Operational semantics is given to $\Sigma^*$-terms---that is, to arrows of the
prop $\syntax^*=\mathbf{S}_{(\Sigma^*,\varnothing)}$---where $\Sigma^*$ is
obtained from set of generators in~\eqref{eq:generators} by replacing the formal
variables $s$ and $s^{-1}$ with a family of registers, indexed by the scalars of
$\k$:
\[
\delaygen \leadsto \{ \delaygenk | \, k\in\k \}
\]
\[
\delayopgen \leadsto \{ \delayopgenk |\, k\in\k \}
\]
The idea is that at any time during a computation the register holds the signal
it has received on the previous `clock-tick'. There are no equations, apart from
the laws of symmetric monoidal categories.

Next we introduce the structural rules: the transition relations that occur at
each clock-tick, turning one $\Sigma^*$-term into another. Each transition is given
two labels, written above and below the arrow.  The upper label refers to the
signals observed at the `dangling wires' on the left-hand side of the term, and
the lower label to those observed on the right hand side.  Indeed, transitions
out of a term of type $m \to n$ have upper labels in $\k^m$ and lower labels in
$\k^n$.

Because---for the purposes of the operational account---we consider these terms
to be syntactic, we must also account for the twist $\twist$ and identity $\id$.
A summary of the structural rules is given below; the rules for the mirror image
generators are symmetric, in the sense that upper and lower labels are swapped.

\[
\copygen \dtrans{\,k\,}{k \, k} \copygen \quad 
\discardgen \dtrans{k}{} \discardgen   
\]
\[
\addgen \dtrans{k \labelSep\, l }{k+l} \addgen \quad 
\zerogen \dtrans{\phantom{b}}{0} \zerogen
\]
\[
\scalargen  \dtrans{\,\,l\,\,}{al} \scalargen \quad 
\delaygenl \dtrans{k}{l} \delaygenk 
\]
\[
\id  \dtrans{k}{k} \id \quad 
\twist \dtrans{k \labelSep l}{l \labelSep k} \twist 
\]
\[
  \frac{s\dtrans{\mathbf{u}}{\mathbf{v}} s' \quad
  t\dtrans{\mathbf{v}}{\mathbf{w}} t'}{s \mathrel{;} t \dtrans
  {\mathbf{u}}{\mathbf{w}} s' \mathrel{;} t'}
 \quad 
 \frac{s\dtrans{\mathbf{u_1}}{\mathbf{v_1}} s'\quad
 t\dtrans{\mathbf{u_2}}{\mathbf{v_2}} t'}
 {s\oplus t \dtrans{\mathbf{u_1 \labelSep u_2}}{\mathbf{v_1 \labelSep v_2}} s'\oplus t'}
\]
Here $k, l,a \in \k$, $s,t$ are $\Sigma^*$-terms, and
$\mathbf{u,v,w,u_1,u_2,v_1,v_2}$
are vectors in $\k$ of the appropriate length. Note that the only generators
that change as a result of computation are the registers; it follows this is the
only source of state in any computation.

The structural rules are identical to those given in~\cite[\S 2]{BSZ3}.
Differently from~\cite{BSZ3}, however, we can be more liberal with our
assumptions about the \emph{initial states} of computations.  In~\cite{BSZ3}
each computation starts off with all registers initialised with the 0($\in\k$)
value.  For us, systems can be initialised with arbitrary elements of $\k$. 

Let $\tm\colon m\to n \in \syntax$ be a $\Sigma$-term. Fixing an ordering of
the delays $\delaygen$ in $\tm$ allows us to identify the set of delays with a
finite ordinal $[d]$.  A \define{register assignment} is then simply a function
$\sigma: [d]\to \k$.  We may instantiate the $\Sigma$-term $\tm$ with the
register assignment $\sigma$ to obtain the $\Sigma^\ast$-term $\tm_\sigma
\in\syntax^\ast$ of the same type: for all $i \in [d]$ simply replace the $i$th
delay with a register in state $\sigma(i)$.

A computation on $\tm$ initialised at $\sigma$ is an infinite sequence of register
assignments and transition relations 
\[
  \tm_\sigma \dtrans{\mathbf{u_1}}{\mathbf{v_1}} \tm_{\sigma_1}
  \dtrans{\mathbf{u_2}}{\mathbf{v_2}} \tm_{\sigma_2}  
  \dtrans{\mathbf{u_3}}{\mathbf{v_3}} \dots
\]
The \define{trace} of this computation is the sequence $\mathbf{(u_1,v_1),
(u_2,v_2), \dots}$ of elements of $\k^m \oplus \k^n$.

To relate these structural rules to our denotational semantics, we introduce the
notion of biinfinite trace. A biinfinite trace is a trace with an infinite past
as well as future. To discuss these, we use the notion of a reverse computation: 
a computation using the operational rules above, but with the rules for delay having their
  left and right hand sides swapped, e.g.
\[
\delaygenk \dtrans{k}{l} \delaygenl.
\]

\begin{definition}
  Given $\tm\colon m \to n \in \syntax$, a \define{biinfinite trace on} $\tm$ is a sequence
$w \in (\k^m)^\z \oplus (\k^n)^\k$ such that there exists
\begin{enumerate}[(i)]
\item a register assignment $\sigma$; 
\item an infinite \emph{forward trace} $\phi_\sigma$ of a computation on $\tm$
  initialised at $\sigma$; and,
\item an infinite \emph{backward trace} $\psi_\sigma$ of a reverse computation
  on $\tm$ initialised at $\sigma$,
\end{enumerate}
obeying
\[
w(t) = \begin{cases}
\phi(t) & t\geq 0; \\
\psi(-(t+1)) & t < 0.
\end{cases}
\]
We write $\bit(\tm)$ for the set of all biinfinite traces on $\tm$.
\end{definition}


The following result gives a tight correspondence between the operational and
denotational semantics, and follows via a straightforward structural induction
on $\tm$.
\begin{lemma}
For any $\tm\maps m\to n \in \syntax$, we have 
\[
  \llbracket \tm \rrbracket = \bit(\tm) 
\]
as subsets of $(\k^m)^\z\times (\k^n)^\z$.
\end{lemma}

\section{Controllability} \label{sec.control}

Suppose we are given a current and a target trajectory for a system. Is it
always possible, in finite time, to steer the system onto the target trajectory? If so, the 
 system is deemed controllable, and the problem of controllability of systems is at
the core of control theory. The following definition is due to
Willems~\cite{Wi2}.

\begin{definition}\label{def:contr}
A system $(T,W,\bb)$ (or simply the behaviour $\bb$) is \define{controllable} if
for all $w,w' \in \bb$ and all times $t_0\in\mathbb{Z}$, there exists $w'' \in
\bb$ and $t_0^\prime\in\mathbb{Z}$ such that $t_0'>t_0$ and $w''$ obeys
\[
  w''(t) = 
  \begin{cases} 
    w(t) & t \le t_0 \\
    w'(t-t_0') & t \ge t_0'.
  \end{cases}
\]
\end{definition}

As mentioned in the introduction, a novel feature of our graphical calculus is
that it permits the representation of non-controllable behaviours.

\begin{exm} \label{ex.noncontrol}
Consider the system in the introduction, represented by the signal flow
graph~\eqref{eq:examplesfg}. The equation 
\[
\includegraphics[height=1.3cm]{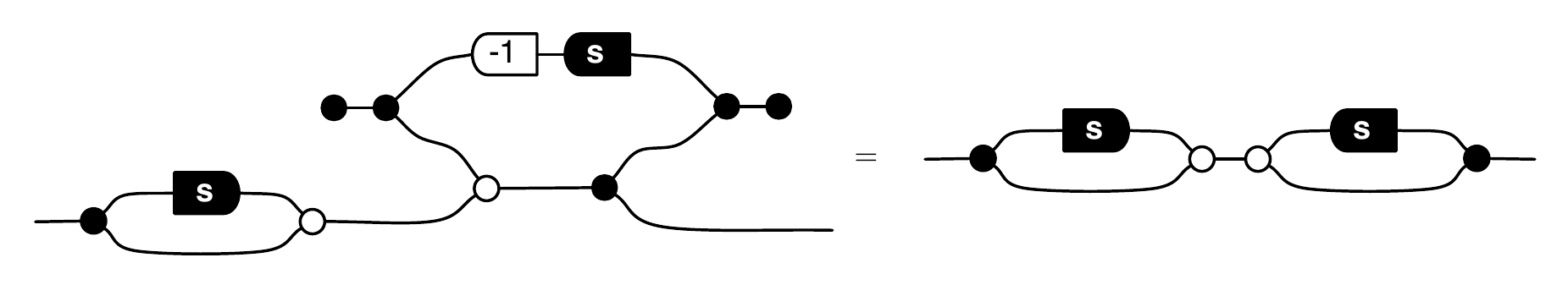}
\]
follows from the equations of $\ltids$, and so as an $\ltids$ system
\eqref{eq:examplesfg} is the system represented by the corelation
$1 \xrightarrow{[s+1]} 1 \xleftarrow{[s+1]} 1.$

The trajectories of this system are precisely those sequences $w= (w_1, w_2) \in
\k^\z \oplus \k^\z$ that satisfy $(s+1)\cdot w_1= (s+1)\cdot w_2$; that is, they
satisfy the difference equation
\[
  w_1(t-1)+w_1(t)-w_2(t-1)-w_2(t)=0.
\]

To see that the system is non-controllable, note that this implies that every
trajectory $w$ obeys
\[
  w_1(t-1)-w_2(t-1) = -(w_1(t)-w_2(t)),
\]
so $(w_1-w_2)(t) = (-1)^tc_w$ for some $c_w \in \k$. This $c_w$ is a time-invariant
property of any trajectory. Thus if $w$ and $w'$ are trajectories such that $c_w \ne
c_{w'}$, then it is not possible to transition from the past of $w$ to the
future of $w'$ along some trajectory in $\mathcal B$. 

Explicitly, taking $w(t) = ((-1)^t,0)$ and $w'(t) = ((-1)^t2,0)$ suffices to
show $\mathcal B$ is not controllable.
\end{exm}

\subsection{A categorical characterisation}
We now show that controllable systems are precisely those representable as
\emph{spans} of matrices. This novel characterisation leads to new ways of
reasoning about controllability of composite systems. 

Among the various equivalent conditions for controllability, the existence of
\emph{image representations} is most useful for our purposes.
\begin{proposition}[Willems {\cite[p.86]{Wi}}] \label{thm.imagereps}
  An LTI behaviour $\bb$ is controllable iff there
  exists $M \in \mat\pk$ such that $\bb = \im \theta M$.
\end{proposition}

Restated in our language, Prop. \ref{thm.imagereps} states that controllable
systems are precisely those representable as \emph{spans} of matrices. 

\begin{theorem} \label{cor.spanreps}
  Let $m \xrightarrow{A} d \xleftarrow{B} n$ be a corelation in $\corel\mat\pk$. Then
  $\Phi(\xrightarrow{A}\xleftarrow{B})$ is controllable iff
  there exist matrices $R: e \to m$, $S: e\to n$ such that 
  \[
    m \xleftarrow{R} e \xrightarrow{S} n = m \xrightarrow{A} d \xleftarrow{B} n
  \]
  as morphisms in $\corel\mat\pk$. 
\end{theorem}
\begin{proof}
  To begin, note that the behaviour of a span is its joint image. That is,
  $\Phi(\xleftarrow{R}\xrightarrow{S})$ is the composite of linear
  relations $\ker\vectfun[\mathrm{id}_m \; -R]$ and $\ker\vectfun[S\;
  -\mathrm{id}_n]$, which comprises all $(\mathbf{x},\mathbf{y}) \in (\k^m)^\z
  \oplus (\k^n)^\z$ such that there exists $\mathbf{z} \in (\k^e)^\z$ with
  $\mathbf{x} = \vectfun R \mathbf{z}$ and $\mathbf{y} = \vectfun S
  \mathbf{z}$. Thus
  \[
    \Phi(\xleftarrow{R}\xrightarrow{S}) = \im \vectfun \left[
    \begin{matrix} R \\ S \end{matrix} \right].
  \]
  The result then follows immediately from Prop. \ref{thm.imagereps}.  
\end{proof}

In terms of the graphical theory, this means that a term of in the canonical
form $\ha_\pk ; \ha_\pk^{op}$ (`cospan form') is controllable iff we can find a
derivation, using the rules of $\ihcor$, that puts it in the form $\ha_\pk^{op}
; \ha_\pk$ (`span form').  This provides a general, easily recognisable
representation for controllable systems. 

The idea of a span representation also leads to a test for controllability: take
the pullback of the cospan and check whether the system described by it
coincides with the original one. Indeed, note that as $\pk$ is a PID, the
category $\mat\pk$ has pullbacks. A further consequence of Th.
\ref{cor.spanreps}, together with Prop.  \ref{prop.magic}, is the following. 

\begin{proposition} \label{prop.ctrlablepart}
  Let $m \xrightarrow{A} d \xleftarrow{B} n$ be a cospan in $\mat\pk$, and write
  the pullback of this span $m \xleftarrow{R} e \xrightarrow{S} n$. Then the
  behaviour of the pullback span $\Phi(\xleftarrow{R}\xrightarrow{S})$ is
  the maximal controllable sub-behaviour of
  $\Phi(\xrightarrow{A}\xleftarrow{B})$.
\end{proposition}
\begin{proof}
  Suppose we have another controllable behaviour $\mathscr{C}$ contained in
  $\ker\vectfun [A\;-B]$. Then this behaviour is the $\Phi$-image of
  some span $m \xleftarrow{R'} e' \xrightarrow{S'}n$. But by the universal
  property of the pushout, there is then a map $b' \to b$ such that the relevant
  diagram commutes.  This implies that the controllable behaviour $\mathscr{C} =
  \im\vectfun\begin{bmatrix} R' \\ S'\end{bmatrix}$ is
  contained in $\im\vectfun \begin{bmatrix} R \\ S\end{bmatrix}$,
  as required. 
\end{proof}

\begin{corollary}
  Suppose that an LTI behaviour $\bb$ has cospan representation
  \[
    m \stackrel{A}\longrightarrow d \stackrel{B}\longleftarrow n.
  \]
  Then $\bb$ is controllable iff the $\Phi$-image of the pullback of this cospan
  in $\mat\pk$ is equal to $\bb$.
\end{corollary}

Moreover, taking the pushout of this pullback span gives another cospan. The
morphism from the pushout to the original cospan, given by the universal
property of the pushout, describes the way in which the system fails to be
controllable.

Graphically, the pullback may be computed by using the axioms of the theory of
interacting Hopf algebras $\ih_\pk$~\cite{BSZ2,Za}. For example, the pullback
span of the system of Ex. \ref{ex.noncontrol} is simply the identity span, as
derived in equation \eqref{eq:exampleproof} of the introduction. In the
traditional matrix calculus for control theory, one derives this by noting the
system has kernel representation $\ker\theta\begin{bmatrix} s+1 & -(s+1)
\end{bmatrix}$, and eliminating the common factor $s+1$ between the entries.
Either way, we conclude that the maximally controllable subsystem of $1
\xrightarrow{[s+1]} 1 \xleftarrow{[s+1]} 1$ is simply the identity system $1
\xrightarrow{[1]} 1 \xleftarrow{[1]} 1$.


\subsection{Control and interconnection}
From this vantage point we can make useful observations about controllable
systems and their composites: we simply need to ask whether we can write them
as spans. This provides a new proof of the following.

\begin{proposition}
  Suppose that $\bb$ has cospan representation
    $m \xrightarrow{A} d \xleftarrow{B} n.$
  Then $\bb$ is controllable if $A$ or $B$ is invertible.
\end{proposition}
\begin{proof}
  If $A$ is invertible, then $m \xleftarrow{A^{-1}} d \xleftarrow{B} n$ is an
  equivalent span; if $B$ is invertible, then $m \xrightarrow{A} d
  \xrightarrow{B^{-1}} n$.
\end{proof}

More significantly, the compositionality of our framework aids understanding of
how controllability behaves under the interconnection of systems---an active
field of investigation in current control theory. We give an example application
of our result.

First, consider the following proposition.
\begin{proposition}\label{prop:veryexciting}
  Let $\bb,\mathscr{C}$ be controllable systems, given by the respective
  $\Phi$-images of the spans
    $m \xleftarrow{B_1} d \xrightarrow{B_2} n$ 
    and 
    $n \xleftarrow{C_1} e \xrightarrow{C_2} l$.
  Then the composite $\mathscr{C} \circ \bb: m \to l$ is controllable
  if $\Phi(\xrightarrow{B_2}\xleftarrow{C_1})$ is
  controllable.
\end{proposition}
\begin{proof}
  Replacing $\xrightarrow{B_2}\xleftarrow{C_1}$ with an equivalent span gives a span
  representation for $\mathscr{C} \circ \bb$.
\end{proof}

\begin{exm}
  Consider LTI systems
  \[
  \lower25pt\hbox{$\includegraphics[height=2cm]{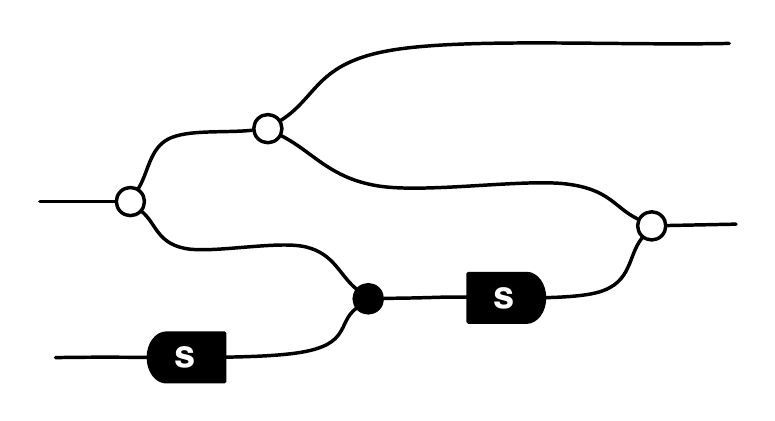}$}
  \quad\text{and}\quad
  \lower22pt\hbox{$\includegraphics[height=1.8cm]{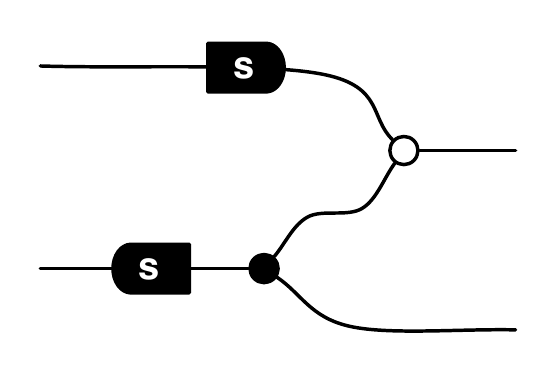}$}.
  \]
  These systems are controllable because each is represented by a span in
  $\mat\pk$. Indeed, recall that each generator of $\ltids=\corel\mat\pk$ arises
  as the image of a generator in $\mat\pk$ or $\mat\pk^{\mathrm{op}}$; for
  example, the white monoid map $\addgen$ represents a morphism in $\mat\pk$,
  while the black monoid map $\copyopgen$ represents a morphism in
  $\mat\pk^{\mathrm{op}}$. The above diagrams are spans as we may partition the
  diagrams above so that each generator in $\mat\pk^{\mathrm{op}}$ lies to the
  left of each generator in $\mat\pk$.

  To determine controllability of the interconnected system
  \[
  \lower25pt\hbox{$\includegraphics[height=2cm]{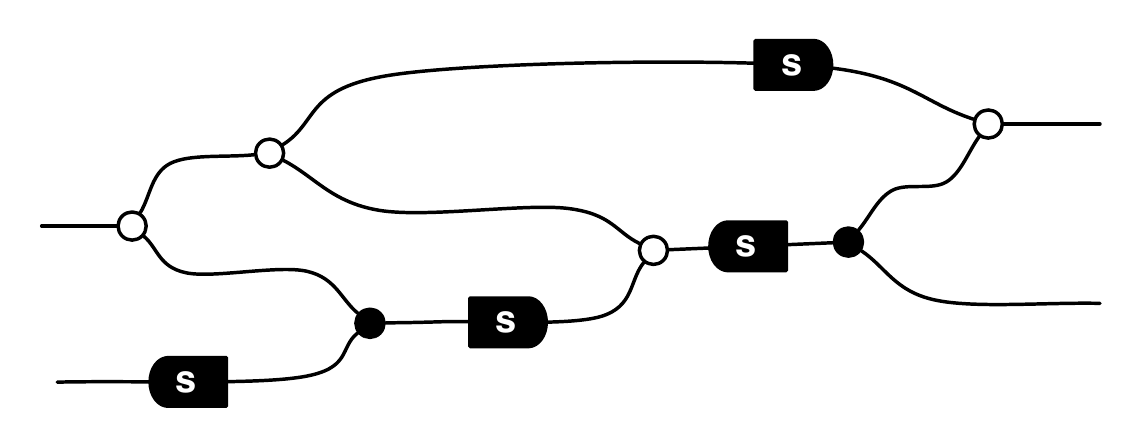}$}
  \]
  Prop. \ref{prop:veryexciting} states that it is enough to consider the
  controllability of the subsystem
  \[
  \lower17pt\hbox{$\includegraphics[height=1.4cm]{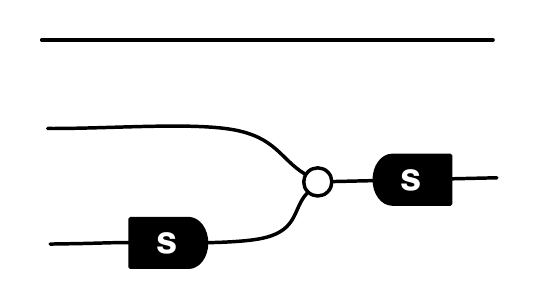}$}.
  \]
  The above diagram gives a representation of the subsystem as a cospan in
  $\mat\pk$. We can prove it is controllable by rewriting it as a span using an
  equation of $\ltids$:
  \[
  \lower17pt\hbox{$\includegraphics[height=1.4cm]{pics/68diag4.pdf}$}
    =
  \lower17pt\hbox{$\includegraphics[height=1.4cm]{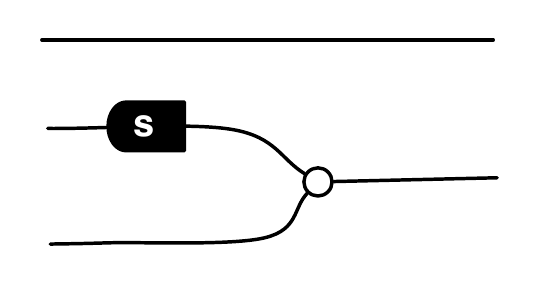}$}.
  \]
  Thus the composite system is controllable.
\end{exm}

\subsection{Comparison to matrix methods}
The facility with which the graphical calculus formalises and solves such
controllability issues is especially appealing in view of potential applications
in the analysis of controllability of \emph{systems over networks} (see
\cite{OFM}). To make the reader fully appreciate such potential, we sketch how
complicated such analysis is using standard algebraic methods and
dynamical system theory for the case of two single-input, single-output systems.
See also pp. 513--516 of \cite{FH}, where a generalization of the result of
Prop. \ref{prop:veryexciting} is given in a polynomial- and operator-theoretic
setting. 

In the following we abuse notation by writing a matrix for its image under the
functor $\theta$.  The following is a useful result for proving the
controllability of kernel representations in the single-input single-output case.

\begin{proposition}[Willems {\cite[p.75]{Wi}}] \label{prop.controlkernel}
  Let $\mathscr B \subseteq (\k^2)^\z$ be a behaviour given by the kernel of the
  matrix $\begin{bmatrix} A & B\end{bmatrix}$, where $A$ and $B$ are column
  vectors with entries in $\k[s,s^{-1}]$. Then $\mathscr B$ is controllable if
  and only if the greatest common divisor $\gcd(A,B)$ of $A$ and $B$ is $1$.
\end{proposition}

Using the notation of Prop. \ref{prop:veryexciting}, the trajectories of
$\mathscr{B}$ and $\mathscr{C}$ respectively are those $(w_1,w_2) \in (\k^n)^\z
\oplus (\k^m)^\z$ and $(w_2',w_3) \in (\k^m)^\z \oplus (\k^p)^\z$ such that 
\begin{equation}\label{eq:imBnC}
\begin{bmatrix} w_1\\w_2 \end{bmatrix}
=
\begin{bmatrix} B_1\\B_2 \end{bmatrix} \ell_1 
\quad \mbox{ \rm and} \quad  
\begin{bmatrix} w_2^\prime\\w_3 \end{bmatrix}
=
\begin{bmatrix} C_1\\C_2 \end{bmatrix} \ell_2\; 
\end{equation}
for some $\ell_1 \in (\k^b)^\z$, $\ell_2 \in (\k^c)^\z$. These are the explicit
image representations of the two systems.  We assume without loss
of generality that the representations (\ref{eq:imBnC}) are \emph{observable}
(see \cite{Wi}); this is equivalent to $\gcd(B_1,B_2)=\gcd(C_1,C_2)=1$. Augmenting
(\ref{eq:imBnC}) with the interconnection constraint
$w_2=B_2\ell_1=C_1\ell_2=w_2^\prime$ we obtain the representation of the
interconnection: 
\begin{eqnarray}\label{eq:hybBnC}
\begin{bmatrix}
w_1\\w_2\\w_2^\prime\\w_3\\0
\end{bmatrix}&=&\begin{bmatrix} B_1&0\\B_2&0\\0&C_1\\ 0&C_2 \\ B_2&-C_1 \end{bmatrix} \begin{bmatrix} \ell_1\\\ell_2\end{bmatrix}\; .
\end{eqnarray}
Prop. \ref{prop:veryexciting} concerns the controllability of the set
$\mathscr{C} \circ \mathscr{B}$ of trajectories $(w_1,w_3)$ for which there
exist trajectories $w_2$, $w_2^\prime$, $\ell_1$, $\ell_2$ such that
(\ref{eq:hybBnC}) holds. 

To obtain a representation of such behavior the variables $\ell_1$, $\ell_2$,
$w_2$ and $w_2^\prime$ must be eliminated from (\ref{eq:hybBnC}) via algebraic
manipulations (see the discussion on p. 237 of \cite{Wi2}). Denote
$G=\gcd(B_1,C_2)$, and write $C_2=G C_2^\prime$ and $B_1=G B_1^\prime$, where
$\gcd(B_1^\prime, C_2^\prime)=1$.  Without entering in the algebraic details, it
can be shown that a kernel representation of the projection of the behavior of
(\ref{eq:hybBnC}) on the variables $w_1$ and $w_3$ is
\begin{equation}\label{eq:extafterelim}
  \begin{bmatrix} C_2^\prime B_2& -B_1^\prime C_2\end{bmatrix} 
  \begin{bmatrix} w_1\\ w_3\end{bmatrix}=0\; .
\end{equation}
We now restrict to the single-input single-output case. Recalling Prop.
\ref{prop.controlkernel}, the behavior represented by (\ref{eq:extafterelim}) is
controllable if and only if $\gcd(C_2^\prime B_2, B_1^\prime C_1)=1$. 

Finally then, to complete our alternate proof of the single-input single-output
case of Prop. \ref{prop:veryexciting}, note that
$\cospanfunrest(\xrightarrow{B_2}\xleftarrow{C_1})$ is controllable if
$\gcd(B_2, C_1)=1$.  Given the observability assumption, this implies
$\gcd(C_2^\prime B_2, B_1^\prime C_1)=1$, and so the interconnected behaviour
$\mathscr{C} \circ \mathscr{B}$ represented by (\ref{eq:extafterelim}) is
controllable. 

In the multi-input, multi-output case stating explicit conditions on the
controllability of the interconnection on the basis of properties of the
representations of the individual systems and their interconnection is rather
complicated. This makes the simplicity of Prop. \ref{prop:veryexciting} and the
straightforward nature of its proof all the more appealing.

\section{Conclusion} \label{sec.close}

Willems concludes~\cite{Wi} with 
\begin{quote}
Thinking of a dynamical system as a behavior, and of interconnection as variable sharing, gets the physics right. 
\end{quote}
In this paper we have shown that the algebra of symmetric monoidal categories gets the mathematics right. 

We characterised the prop $\ltids$ of linear time invariant dynamical systems as the prop of corelations of matrices over $\k[s,s^{-1}]$ and used this fact to present it as a symmetric monoidal theory. As a result, we obtained the language of string diagrams as a syntax for LTI systems. From the point of view of formal semantics, the syntax was endowed with denotational and operational interpretations that are closely related, as well as a sound and complete system of equations for diagrammatic reasoning. 

We harnessed the compositional nature of the language to provide a new characterisation of the fundamental notion of controllability, and argued that this approach is well-suited to some of the problems that are currently of interest in systems theory, for example in systems over networks, where compositionality seems to be a vital missing ingredient. 

The power of compositionality will be of no surprise to researchers in concurrency theory or formal semantics of programming languages. Our theory--which, as we show in the final section, departs radically from traditional techniques---brings this insights to control theory. By establishing links between our communities, our ambition is to open up control theory to formal specification and verification.


%


\begin{appendix}

\newpage 

\section{Background on cospans} \label{app.cospans}
Suppose $\mathcal C$ is a category with pushouts. A \define{cospan} from $X$ to
$Y$ in $\mathcal C$ is an object $N$ with a pair of morphisms $X\xrightarrow{i}
N \xleftarrow{o}Y$.
Given cospans $X\xrightarrow{i_X} N \xleftarrow{o_Y}Y$, $Y\xrightarrow{i_Y} M
\xleftarrow{o_Z}Z$, we may compose them by taking a pushout over $Y$:
\[
  \xymatrix{
    X \ar[dr]_{i_X} & & Y \ar[dl]^{o_Y}  
    \ar[dr]_{i_Y} & & Z \ar[dl]^{o_Z} {} 
    \\
    & N \ar[dr]_{\iota_N} & & M \ar[dl]^{\iota_M}  \\
    & & P {\save*!<0cm,-.5cm>[dl]@^{|-}\restore}
  }
\]
The composite is the cospan $X\xrightarrow{i_X;\iota_N} P
\xleftarrow{o_Z;\iota_M}Z$.

Given cospans $X\xrightarrow{i} N \xleftarrow{o}Y$, $X\xrightarrow{i'} N'
\xleftarrow{o'}Y$, a \define{map of cospans} is a morphism $n\colon  N \to N'$ in $\mathcal C$
such that
\[
  \xymatrix{
    & N \ar[dd]^n  \\
    X \ar[ur]^{i} \ar[dr]_{i'} && Y \ar[ul]_{o} \ar[dl]^{o'}\\
    & N'
  }
\]
commutes. We may define a category $\cospan \mathcal{C}$ with objects the
objects of $\mathcal C$ and morphisms isomorphism classes of cospans \cite{Be}.

\section{Background on factorisation systems}\label{app.factorisation}


Recall that a \define{factorisation system} $(\mathcal L,\mathcal R)$ in a category
$\mathcal C$ comprises subcategories $\mathcal L$, $\mathcal R$ of $\mathcal C$
such that 
\begin{enumerate}[(i)]
  \item $\mathcal L$ and $\mathcal R$ contain all isomorphisms of $\mathcal C$
  \item every morphism $f$ of $\mathcal C$ can be factored as $f = l;r$,
    where $l \in \mathcal L$ and $r \in \mathcal R$ 
  \item given morphisms $f,f'$, with factorisations $f = l;r$, $f' = l';r'$ of the
    above sort, for every $u$, $v$ such that $f;v = u;f'$ there exists a unique
    morphism $t$ such that
    \[
      \xymatrixcolsep{3pc}
      \xymatrixrowsep{3pc}
      \xymatrix{
  \ar[r]^l \ar[d]_u & \ar[r]^r \ar@{.>}[d]^{\exists! t} &  \ar[d]^v \\
  \ar[r]_{l'}& \ar[r]_{r'} & 
      }
    \]
    commutes.
\end{enumerate}

Also recall that we say that a subcategory $\mathcal R$ of $\mathcal C$ is
\define{stable under pushouts} if whenever
\[
  \xymatrixcolsep{3pc}
  \xymatrixrowsep{3pc}
  \xymatrix{
    \ar[r]^{\hat r} & \\
    \ar[u] \ar[r]^{r} & \ar[u]
  }
\]
is a pushout square in $\mathcal C$ such that $r \in \mathcal R$, we also have
$\hat r \in \mathcal R$.


\section{Proofs} \label{app.proofs}


\subsection{Proof of Proposition~\ref{prop.matfactorisation}\label{app.prop.matfactorisation}}
\begin{proof}
  We use Smith normal form~\cite[Section~6.3]{Ka}. Given any matrix $R$, the Smith normal form gives
  us  $R = VDU$, where $U$ and $V$ are invertible, and $D$
  is diagonal. In graphical notation we can write it thus:
  \[
\includegraphics[height=1.2cm]{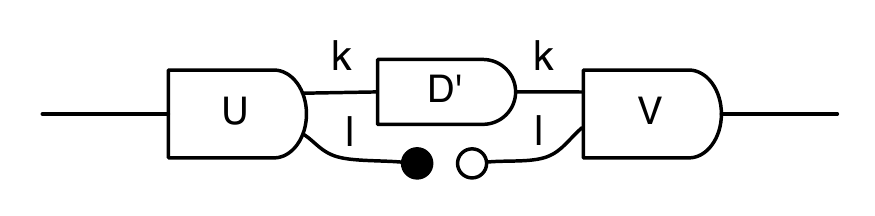}
  \]
  This implies we may write it as $R = U';D';V'$, where $U'$ is a split
  epimorphism, $D'$ diagonal of full rank, and $V'$ a split monomorphism.
  Explicitly, the construction is given by
  \[
\includegraphics[height=2cm]{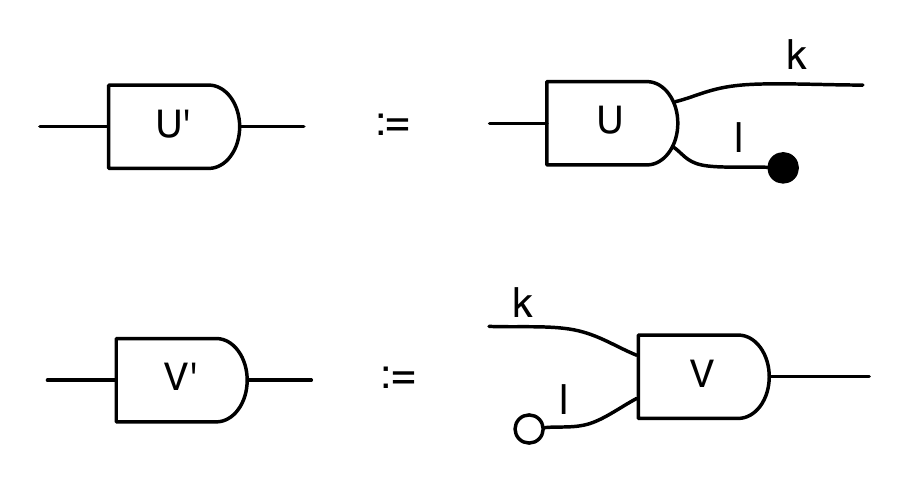}
  \]
  Recall that $\pk$ is a PID, so the full rank diagonal matrix $D'$ is epi. Then
  $R = V'(D'U')$ is an epi-split mono factorisation. 
\end{proof}

\subsection{Proof of Corollary~\ref{cor.episplitmono}\label{app.cor.episplitmono}}
\begin{proof}
That $(\mathcal E, \mathcal M)$ is a factorisation system follows from a more
general fact, true in any category, that if every arrow can be factorised as an
epi followed by a split mono, then this results in a factorisation system.  The
only non-trivial part to check is the uniqueness condition: given epi $e_1,e_2$,
split mono $m_1,m_2$ and commutative diagram
    \[
      \xymatrixcolsep{3pc}
      \xymatrixrowsep{3pc}
      \xymatrix{
  \ar[r]^{e_1} \ar[d]_u & \ar[r]^{m_1} \ar@{.>}[d]^{\exists! t} &  \ar[d]^v \\
  \ar[r]_{e_2}& \ar[r]_{m_2} & 
      }
    \]
we must show that there is a unique $t$ that makes the diagram commute.
Indeed let $t\Defeq m_2'vm_1$ where $m_2'$ satisfies $m_2'm_2=id$. 
To see that the right square commutes, first
\[
m_2 t e_1 =  m_2 m_2' v m_1 e_1 = m_2 m_2' m_2 e_2 u = m_2 e_2 u = v m_1 e_1
\]
and since $e_1$ is epi we have $m_2 t = v m_1$. For the left square,
$t e_1 = m_2' v m_1 e_1 = m_2' m_2 e_2 u = e_2 u$. Uniqueness is immediate,
since, e.g.\ $m_2$ is mono. 

To prove that $\mathcal M$ is stable under pushouts, we make use of the fact
that $\mat\pk$ is equivalent to $\fmod \pk$, and use an explicit formula for the
pushout in $\fmod \pk$. Let 
\[
  \xymatrixcolsep{3pc}
  \xymatrixrowsep{3pc}
  \xymatrix{
    d \ar[r]^{\hat M} & p\\
    n \ar[u]^A \ar[r]^{M} & e \ar[u]
  }
\]
be a pushout square in $\fmod \pk$, where $M$ is a split mono. Write $M'\colon e
\to n$ for the map such that $M'M = \idn$. We show that $\hat M$ also splits.

First, note that the pushout $p$ is given by $(d\oplus e)/\im [\begin{smallmatrix} A \\ -M
\end{smallmatrix}]$, with $\hat M$ mapping each $v \in e$ to the equivalence
class $\overline v$; we abuse notation to write $v$ for both elements of $e$ and
their image in $d \oplus e$ under the coproduct injection, and similar write $u$
for elements of $d$ and their images in $d \oplus e$. Define $\hat M'\colon p
\to e$ by $\overline{u+v}$ maps to $u+AM'v$. Should this be well-defined, it is clear
that $\hat M'\hat M = \idn$. 

We check that $\hat M'$ is well-defined. Suppose $\overline{u'+v'} =
\overline{u+v}$, $u,u' \in
d$, $v,v' \in e$. Then $u-u'+v-v' \in \im [\begin{smallmatrix} A \\ -M
\end{smallmatrix}]$. This implies there exists $w \in n$ such that $u-u'=Aw$,
$v-v' = -Mw$. Thus
\begin{align*}
  u+AM'v &= (u'-Aw)+AM'(v'-Mw) \\
  &= u'-Aw+AM'v'+AM'Mw \\
  &= u'+AM'v'. 
\end{align*}
This shows that $\hat M$ is a split mono, as required.
\end{proof}

\subsection{Proof that $\corel \vect_\k \cong \linrel_\k$\label{app.corellinrel}}
\begin{proof}
  We show that the map 
  \[
    \corel \vect_\k \longrightarrow \linrel_\k
  \]
  sending each vector space to itself and each jointly-epic cospan
  \[
    U \xrightarrow{f} A \xleftarrow{g} V
  \]
  to the linear subspace $\ker[f\;-g]$ is a full, faithful, and
  bijective-on-objects functor.

  Indeed, (equivalence classes of) jointly-epic cospans $U \xrightarrow{f} A \xleftarrow{g} V$ are in
  one-to-one correspondence with surjective linear maps $U\oplus V \to A$, which
  are in turn, by the isomorphism theorem, in one-to-one correspondence with
  subspaces of $U\oplus V$. These correspondences are described by the kernel
  construction above. Thus our map is evidently full, faithful, and
  bijective-on-objects. It also maps identities to identities. It remains to
  check that it preserves composition.

  Suppose we have jointly epic cospans $U \xrightarrow{f} A \xleftarrow{g} V$
  and $V \xrightarrow{h} B \xleftarrow{k} W$. Then their pushout is given by
  $P=A \oplus B/\ker[g\;-h]$, and we may draw the pushout diagram
  \[
    \xymatrix{
      U \ar[dr]_{f} & & V \ar[dl]^{g}  
      \ar[dr]_{h} & & W \ar[dl]^{k} {} 
      \\
      & A \ar[dr]_{\iota_A} & & B \ar[dl]^{\iota_B}  \\
      & & P {\save*!<0cm,-.5cm>[dl]@^{|-}\restore}
    }
  \]
  We wish to show the equality of relations
  \[
    \ker[f\;-g];\ker[h\;-k] = \ker[\iota_A f\; -\iota_B g].
  \]
  Now $(\mathbf{u},\mathbf{w}) \in U \oplus W$ lies in the composite relation
  $\ker[f\;-g];\ker[h\;-k]$ iff there exists $\mathbf{v} \in V$ such that
  $f\mathbf{u} = g\mathbf{v}$ and $h\mathbf{v} = k\mathbf{w}$. But as $P$ is the
  pushout, this is true iff 
  \[
    \iota_A f \mathbf{u} = \iota_A g \mathbf{v} = \iota_B h \mathbf{v} =
    \iota_B k \mathbf{w}.
  \]
  This in turn is true iff $(\mathbf{u}, \mathbf{w}) \in \ker[\iota_Af\;
  -\iota_Bk]$, as required. 
\end{proof}


\onecolumn


\section{Equations of $\ihcor$} \label{app.equations}
\[
  \includegraphics[height=.95\textheight]{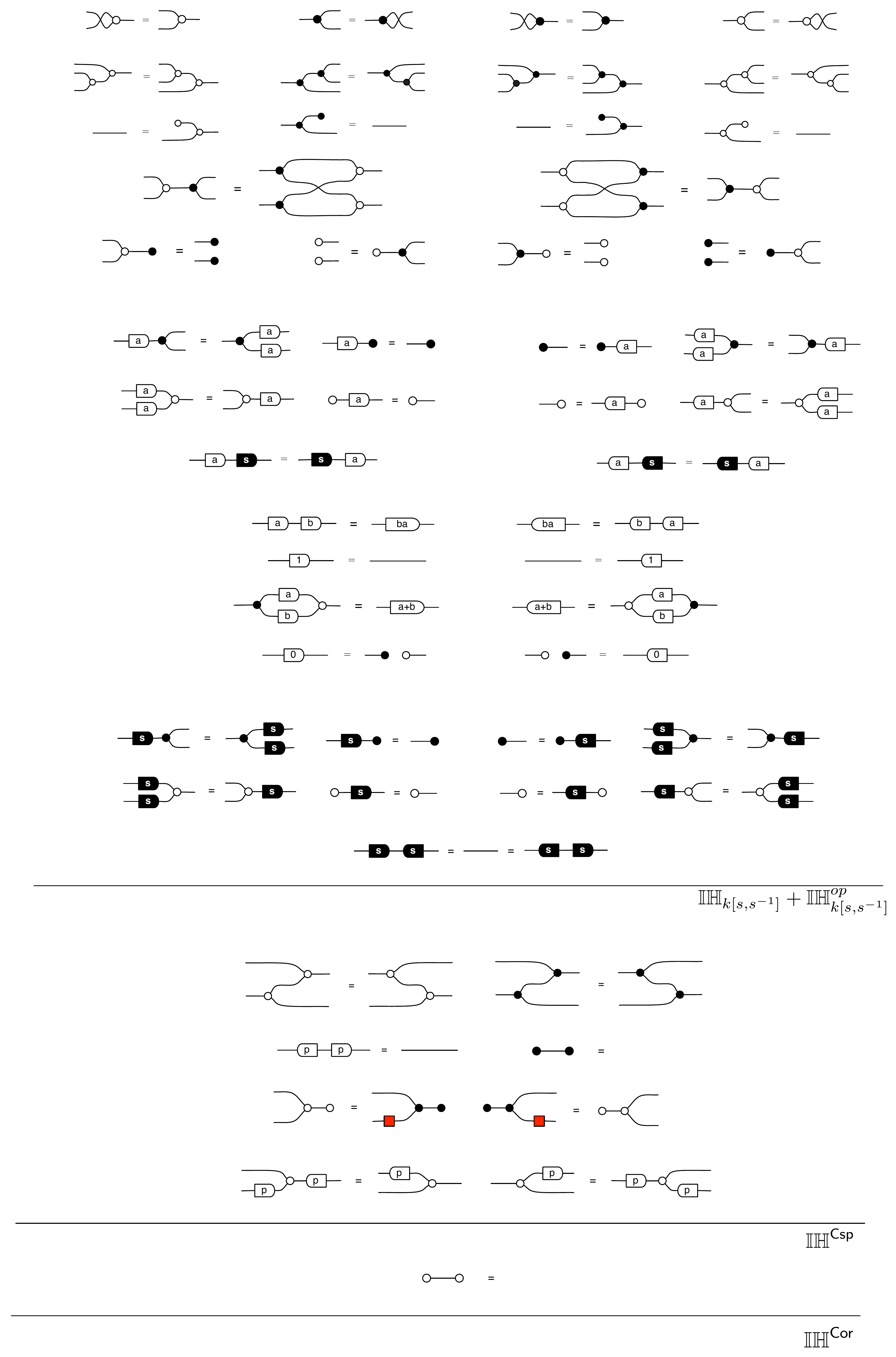}
\]

\end{appendix}

\end{document}